%% file: parityLmain.tex
\pdfoutput=1

\documentclass[11pt]{article}
\usepackage{amsmath,amsfonts,amssymb}
\input{macros1} 
\usepackage{newalg}

\renewcommand{\labelenumi}{(\roman{enumi})}

\begin{document}

\begin{titlepage}
  \vspace*{\fill}
  \begin{center}
    \begin{huge}      { Formal Theories for Logspace Counting }    \end{huge}  \\
    \vspace{1in}
    \begin{Large}      Lila A. Fontes    \end{Large}  \\
    \vspace{1in}
    Advisor: Stephen A. Cook \\
    \vspace{.75in}
    {\sc Research paper submitted in partial fulfillment \\
      of the requirements for the degree of \\
      Master of Science \\ 
    \vspace{.75in}
    Department of Computer Science \\
    University of Toronto \\
    \vspace{.75in}
    Toronto, Ontario \\
    March 19, 2009 }
  \end{center}
  \vspace{1in}
\end{titlepage}

\begin{preliminary}

\tableofcontents

\include{parityLack}

\end{preliminary}

\include{parityLintro}


\input{parityLnotation}

\input{parityLformalizing}


\input{parityLmod_2L}

\input{parityLnumberL}

\input{parityLfuturework}


\include{parityLsources}

\end{document}

%% file: macros1.tex
\newcommand{\titlebox}[6]{
    \begin{center}
        \framebox{
            \vbox{
            \hbox to 5.78in { {\bf #1} \hfill #2}
            \vspace{4mm}
            \hbox to 5.78in { {\Large \hfill #3 \hfill} }
            \vspace{2mm}
            \hbox to 5.78in { {\it #4 \hfill Professor: #5} }
        }
    }
    \end{center}
}

\newcounter{theoremC} 

\newtheorem{theorem}[theoremC]{Theorem}

\newtheorem{definition}[theoremC]{Definition}

\newtheorem{lemma}[theoremC]{Lemma}

\newtheorem{corollary}[theoremC]{Corollary}

\newtheorem{claim}[theoremC]{Claim}

\newtheorem{remk}[theoremC]{Remark}

\newtheorem{exmp}[theoremC]{Example}

\newtheorem{Rabin-algorithm}[theoremC]{Algorithm}

\newenvironment{remark}{\begin{remk}
\begin{normalfont}}{\end{normalfont}
\end{remk}}

\def\FullBox{\hbox{\vrule width 8pt height 8pt depth 0pt}}

\def\qed{\ifmmode\qquad\FullBox\else{\unskip\nobreak\hfil
\penalty50\hskip1em\null\nobreak\hfil\FullBox
\parfillskip=0pt\finalhyphendemerits=0\endgraf}\fi}

\def\qedsketch{\ifmmode\Box\else{\unskip\nobreak\hfil
\penalty50\hskip1em\null\nobreak\hfil$\Box$
\parfillskip=0pt\finalhyphendemerits=0\endgraf}\fi}

\newenvironment{proof}{\begin{trivlist} \item {\bf Proof:~~}}
  {\qed\end{trivlist}}

\newenvironment{proofof}[1]{\begin{trivlist} \item {\bf Proof of 
#1:~~}}
  {\qed\end{trivlist}}

\newcommand{\N}{{\mathbb{N}}}

\newcommand{\Z}{{\mathbb Z}}

\newcommand{\zo}{\{0,1\}}

\DeclareMathOperator{\comp}{-COMP}
\DeclareMathOperator{\ind}{-IND}
\DeclareMathOperator{\minaxiom}{-MIN}

\newenvironment{preliminary}%
  {\pagestyle{plain}\pagenumbering{roman}}
  {\cleardoublepage\pagenumbering{arabic}}

\newcommand{\parityL}{{\oplus L}}

\newcommand{\seq}{\operatorname{\it seq}}
\newcommand{\Pair}{\operatorname{\it Pair}}
\newcommand{\Row}{\operatorname{\it Row}}
\newcommand{\Rowtwo}{\operatorname{\it Row_2}}
\newcommand{\entry}{\operatorname{\it entry}}
\newcommand{\Powtwo}{\operatorname{\it Pow}_2}
\newcommand{\PowSeqtwo}{\operatorname{\it PowSeq}_2}
\newcommand{\PowSeqtwostar}{\operatorname{\it PowSeq}_2^\star}
\newcommand{\Prodtwo}{\operatorname{\it Prod}_2}

\newcommand{\numones}{\operatorname{\it numones}}
\newcommand{\bin}{\operatorname{\it bin}}

\newcommand{\Carry}{\operatorname{\it Carry}} 
\newcommand{\Borrow}{\operatorname{\it Borrow}}
\newcommand{\Convert}{\operatorname{\it Convert}}
\newcommand{\intsize}{\operatorname{\it intsize}}
\newcommand{\PowZ}{\operatorname{\it Pow}_\Z}
\newcommand{\PowSeqZ}{\operatorname{\it PowSeq}_\Z}
\newcommand{\PowSeqZstar}{\operatorname{\it PowSeq}_\Z^\star}
\newcommand{\ProdZ}{\operatorname{\it Prod}_\Z}
\newcommand{\Sum}{\operatorname{\it Sum}}
\newcommand{\Sumz}{\operatorname{\it Sum_\Z}}

%% file: parityLack.tex
\thispagestyle{empty}

\section*{Acknowledgements}
\addcontentsline{toc}{section}{Acknowledgements}

\bigskip

I would like to thank my advisor, Stephen A. Cook, for 
suggesting this project, and for his colossal help, enthusiasm, 
and advising over the course of my research.  This project would not
have been possible without the foundations he and Phuong Nguyen layed
in their book {\it Logical Foundations of Proof Complexity}
\cite{CookNguyen}, and without months of conversations as I 
digested that material and began to see how I could build this small
extension of it.

\bigskip

Thanks also to Eric Allender and Michael Soltys, who provided 
insights into the problem of $AC^0$-closure and reductions for the
class $DET$, and pointed towards many useful references on the
same. 

\bigskip

Many thanks for a lifetime of support and encouragement 
are due to my parents and grandparents.

\bigskip

\noindent March 19, 2009 \hfill Lila A. Fontes

%% file: parityLintro.tex
\section{Introduction}
\label{introduction} 

This paper follows the framework of Chapter 9 of Cook and Nguyen's
basic monograph on proof complexity \cite{CookNguyen}.  
Therein, the authors establish a general method for constructing a
two-sorted logical theory that formalizes reasoning using concepts
from a given complexity class.
This logical theory extends their base theory $V^0$ for $AC^0$ by the
addition of a single axiom.
The axiom states the existence of a solution for a complete
problem for a closed complexity class.
Here, completeness and closure are with respect to $AC^0$-reductions,
and in order to use the methods presented in Chapter 9 and earlier
chapters, we will have to establish that our classes are closed under
such reductions. 

We focus on the classes $\# L$ and $\parityL$, logspace
counting and parity (counting mod $2$), respectively.  
Obviously, $L \subseteq \parityL \subseteq \#L$.  In fact, $\#L$
contains the class $\operatorname{MOD}_kL$ for every $k$, and $\#L
\subseteq NC^2 \subseteq P$.  
The logspace counting class $\#L$ is defined from $L$ by 
analogy to $\#P$ and $P$.  
Both counting classes have nice 
complete problems: the problem of computing the permanent of a 
matrix is complete for $\#P$, and the problem of computing the 
determinant of a matrix is complete for $\#L$.
However,
$\#L$ is not known to be closed under several standard notions of 
reduction.  The logspace counting hierarchy $\#LH$ is defined as 
$\#LH_1 = \#L$, $\#LH_{i+1} = \#L^{\#LH_i}$, and $\#LH =
\bigcup_i \#LH_i$.  There are several names for the closure
under consideration: 
$$AC^0(\#L) = \#LH = DET$$
It is unknown whether this hierarchy collapses, though Allender 
shows that $AC^0(\#L) = NC^1(\#L)$ is a sufficient condition 
for this collapse \cite{All04}. 

Since the class $\# L$ is not known to be closed under
$AC^0$-reductions, we consider 
its $AC^0$ closure $AC^0(\#L)$, denoted $DET$.
In renaming this class, we follow the notation of \cite{All04},
\cite{MV}, and others.  
This class is suggestively named: one of its $AC^0$-complete
problems is that of finding the determinant of an integer-valued
matrix. 
(This should not be confused with the notation from Cook's survey
\cite{Cook}, which considers $NC^1$-reductions.)
$\parityL$ is appropriately closed, and the fact that function values
are mod $2$ allows for a convenient notational shortcut not
available in $\# L$: each number can be stored in just 
one bit of a bit string.  
Section \ref{section:DET} develops the theory for $\#L$.
Section \ref{section:2L} develops the theory for $\parityL$.

Following the format of Chapter 9 of \cite{CookNguyen}, a class is
defined in each section 
and shown to be closed under $AC^0$-reductions.  
Next, a complete problem is demonstrated and formalized as an axiom.
(Because the classes are closely related, we use the same problem ---
matrix powering over the appropriate ring --- for both.)  
$VTC^0$ extended by this axiom forms the theory $V\#L$ (respectively,
$V^0(2) \subset V{{\parityL}}$) with vocabulary $\mathcal{L}_A^2$, the
basic language of 
2-sorted arithmetic.  
Following this, we develop a universal conservative extension
$\overline{V\#L}$ (resp. $\overline{V{\parityL}}$), in which every
string function has a symbol, with extended language
$\mathcal{L}_{F\#L}$ (resp. $\mathcal{L}_{F{\parityL}}$).  
By general results from Chapter 9 of \cite{CookNguyen}, the provably total functions of
$V\#L$ are exactly the functions of class $\#L$ (and similarly for
${\parityL}$).  
These functions
include many standard problems of linear algebra over the rings $\Z$
and $\Z_k$ (\cite{BDHM}, \cite{Cook}, \cite{BvGH}, \cite{Berkowitz}),
as discussed in Section \ref{section:future-work}.

%% file: parityLnotation.tex
\section{Notation} \label{notation}

In this section, we restate useful number and string functions from
\cite{CookNguyen}.  The conventional functions with which
we manipulate bit-strings and numbers are so pervasive in the
discussion to follow that they merit their own notation.  We also
extend these functions with new notation, which will be helpful in the
presentation of formal theories in later sections.

Chapter 4 of \cite{CookNguyen} defines the two-sorted first-order
logic of our theories.  There are two kinds of variables (and
predicates and functions): {\it number}
variables, indicated by lowercase letters $x$, $y$, $z$, \ldots, and
{\it string} variables, indicated by uppercase letters  $X$, $Y$, $Z$,
\ldots.  Strings can be 
interpreted as finite subsets of $\N$.  Predicate symbols $P$, $Q$,
$R$, \ldots, can take arguments of both sorts, as can function
symbols. 
Function symbols are differentiated by case: $f$, $g$, $h$, \ldots,
are number functions, and $F$, $G$, $H$, \ldots, are string
functions.

\subsection{Strings encoding matrices and lists}
\label{subsection:str-mat-lists}

Section 5D of \cite{CookNguyen} provides useful notation for encoding a 
$k$-dimensional bit array in a string $X$.  For our purposes below, it
is useful to extend this notation, so that a string $X$ may encode a
$2$-dimensional array of strings.

\cite{CookNguyen} defines the $AC^0$ functions $\operatorname{\it left}$, $\operatorname{\it right}$, $\seq$,
and $\langle \cdot, \cdot \rangle$, and the $AC^0$ relations $\Pair$
and $\Row$, as follows.

The pairing function $\langle x,y\rangle$ is defined as
$$ \langle x , y \rangle  = (x+y)(x+y+1) + 2y $$
This function can be chained to ``pair'' more than two numbers:
$$ \langle x_1, x_2, \ldots, x_k \rangle = 
\langle \langle x_1, \ldots, x_{k-1}\rangle, x_k\rangle$$
Inputs to the pairing function can be recovered using the projection
functions $\operatorname{\it left}$ and $\operatorname{\it right}$:
$$ y = \operatorname{\it left}(x) \leftrightarrow \exists z \leq x (x = \langle
y,z\rangle)
\hspace{1cm}
z=\operatorname{\it right}(x) \leftrightarrow \exists y \leq x (x = \langle y,z\rangle)$$
By definition, $\operatorname{\it left}(x)=\operatorname{\it right}(x) = 0$ when $x$ is not a pair number,
i.e., when $\neg \Pair(x)$, where $\Pair(x) \equiv \exists y,z \leq x (x
= \langle y,z\rangle)$.

Thus a $k$-dimensional bit array can be encoded in string $X$ by:
$$ X(x_1,\ldots,x_k) = X(\langle x_1,\ldots, x_k\rangle)$$

The $\Row$ function is bit-defined as:
$$ \Row(x,Z)(i) \leftrightarrow i < |Z| \wedge Z(x,i)$$
For notational convenience, we write $\Row(x,Z) = Z^{[x]}$.  This can be used to
encode a $1$-dimensional array of $j$ strings $X_1, \ldots, X_j$ in a
single string $Z$, where $X_i = Z^{[i]}$.

It will be useful in Section \ref{Row_2_usage} to compose the $\Row$
function so that a single string encodes 
a 2-dimensional array of strings $X_{i,j}$.  First encode each row $i$
as a string $Y_i$, where $Y_i^{[j]} = X_{i,j}$.  Then encode the list
of strings $Y_i$ as a $1$-dimensional array of strings $Z^{[i]} =
Y_i$.  The resultant string encodes a $2$-dimensional array of
strings.
Let $\Rowtwo(x,y,Z) = Z^{[x][y]}$ represent the string in the
$(x,y)^\textrm{th}$ position of the matrix of strings encoded in $Z$.
$\Rowtwo$ has bit definition:
\begin{equation} \label{Row2}
\Rowtwo(x,y,Z)(i) \leftrightarrow i<|Z| \wedge \Row(x,Z)(y,i)
\end{equation}

Using the $\Row$ Elimination Lemma (5.52 in \cite{CookNguyen}), we can obtain a
$\Sigma_0^B(\mathcal{L}_A^2)$ formula provably equivalent in
$V^0(\Row)$ to (\ref{Row2}).  By Corollary 5.39 in \cite{CookNguyen}, $V^0(\Rowtwo)$ is a
conservative extension of $V^0$.  Also, using the
$\Sigma_0^B$-Transformation Lemma (5.40 in \cite{CookNguyen}), we can obtain an analogous
$\Rowtwo$ elimination lemma.

\begin{lemma} 
  \label{lemma:Row_2_elimination} 
For every $\Sigma_0^B(\Rowtwo)$ formula $\varphi$, there is a
$\Sigma_0^B(\mathcal{L}_A^2)$ formula $\varphi^+$ such that
$V^0(\Rowtwo) \vdash \varphi \leftrightarrow \varphi^+$.
\end{lemma}

Just as $\Row$ is used to extract a list of strings $Z^{[0]},Z^{[1]},
\ldots$, from $Z$, the number function $\seq$ enables string $Z$ to
encode a list of numbers $y_0$, $y_1$, $y_2$, \ldots, where $y_i =
\seq(i,Z)$.  We denote $\seq(i,Z)$ as $(Z)^i$.  The number function
$\seq(x,Z)$ has the defining axiom
$$ y = \seq(x,Z) \leftrightarrow
(y< |Z| \wedge Z(x,y) \wedge \forall z<y, \neg Z(x,z)) \vee
(\forall z<|Z|, \neg Z(x,z) \wedge y=|Z|)$$


It will be useful in Section \ref{section:formalizing} to compose
$\Row$ and $\seq$ so that a string encodes a matrix of numbers.  Each
row of the matrix is encoded as a string $Y_i$, where $(Y_i)^j =
\seq(j, Y_i)$ is the $j^\textrm{th}$ number in the row.  These rows
are 
encoded as a string $Z$, where $Y_i = Z^{[i]} = \Row(i,Z)$.  The
$(i,j)^\textrm{th}$ entry of the matrix is recoverable as
$\entry(i,j,Z)$. 
Number function $\entry$ has the definition
\begin{equation} \label{equation:entry-def}
  \entry(i,j,Z)=y \leftrightarrow (Z^{[i]})^j = y
\end{equation}
The next lemma follows from the $\Row$ Elimination Lemma (5.52 in \cite{CookNguyen}), the
$\Sigma_0^B$-Transformation Lemma (5.40 in \cite{CookNguyen}), and
(\ref{equation:entry-def}).
\begin{lemma}
For every $\Sigma_0^B(\entry)$ formula $\varphi$, there is a
$\Sigma_0^B(\mathcal{L}_A^2)$ formula $\varphi^+$ such that \\
$V^0(\entry) \vdash \varphi \leftrightarrow \varphi^+$.
\end{lemma}

\subsection{Terminology and conventions for defining string functions}

For some string functions, we are only concerned with certain bits of
the output; we ignore all other bits.  However, every bit of output
must be specified in order to define a function and argue about its
uniqueness. The following convention allows us to define a string
function by specifying only its ``interesting'' bits, and requiring
that all other bits be zero.

For example, let $F(i,\vec{x},\vec{X}) = Z$ be a
string function, and let $\varphi$ be its bit-graph:
$$ F(i,\vec{x},\vec{X})(b) \leftrightarrow
\varphi(i,b,\vec{x},\vec{X})$$
Let $G$ be a string function such that $G(\vec{x},\vec{X})^{[i]} =
F(i,\vec{x},\vec{X})$.

The technically correct bit-definition of $G$ must specify every bit
of the output:
\begin{equation} \label{technicalBitDef}
G(\vec{x},\vec{X})(b) \leftrightarrow
\exists i,j<b, \langle i,j\rangle =b \wedge
\varphi(i,j,\vec{x},\vec{X}) 
\end{equation}
For conciseness, throughout this document, such bit-definitions will
instead be written as:
$$G(\vec{x},\vec{X})(i,j) \leftrightarrow
\varphi(i,j,\vec{x},\vec{X})$$ This has the intended meaning of
(\ref{technicalBitDef}), that is, the string function $G$ is false at
all bits $b$ which are \emph{not} pair numbers.  Notice that this
encoding of the list $G$ of strings $F(0,\vec{x},\vec X)$, $F(1,\vec
x,\vec X)$,\ldots,  results in many ``wasted'' bits.

The following definitions are restated here for ease of reference.

\begin{definition}[\cite{CookNguyen} 5.26, Two-Sorted Definability]
\label{def:two-sorted-definability}
Let $\mathcal{T}$ be a theory with vocabulary $\mathcal{L}\supseteq
\mathcal{L}_A^2$, and let $\Phi$ be a set of $\mathcal{L}$-formulas.
A number function $f$ not in $\mathcal{L}$ is $\Phi$-definable in
$\mathcal{T}$ if there is a formula $\varphi(y,\vec x, \vec X)$ in
$\Phi$ such that
$$ \mathcal{T} \vdash \forall \vec x \forall \vec X \exists ! y,
\varphi(y,\vec x, \vec X)$$
and
$$ y = f(\vec x, \vec X) \leftrightarrow \varphi(y,\vec x, \vec X)$$
A string function $F$ not in $\mathcal{L}$ is $\Phi$-definable in
$\mathcal{T}$ if there is a formula $\varphi(\vec x, \vec X, Y)$ in
$\Phi$ such that
$$ \mathcal{T} \vdash \forall \vec x \forall \vec X \exists ! Y,
\varphi(\vec x, \vec X, Y)$$
and
$$ Y = F(\vec x, \vec X) \leftrightarrow \varphi(\vec x, \vec X, Y) $$ 
\end{definition}

\begin{definition}[\cite{CookNguyen} 5.31, Bit-Definable Function]
Let $\varphi$ be a set of $\mathcal{L}$ formulas where $\mathcal{L}
\supseteq \mathcal{L}_A^2$.  We say that a string function symbol
$F(\vec x, \vec X)$ not in $\mathcal{L}$ is $\Phi$-bit-definable from
$\mathcal{L}$ if there is a formula $\varphi(i,\vec x, \vec X)$ in
$\Phi$ and an $\mathcal{L}_A^2$ number term $t(\vec x, \vec X)$ such
that the bit graph of $F$ satisfies
$$ F(\vec x, \vec X)(i) \leftrightarrow i<t(\vec x, \vec X) \wedge
\varphi(i,\vec x, \vec X)$$
The right-hand side of the above equation is the ``bit-definition'' of
$F$. 
\end{definition}

\begin{definition}[\cite{CookNguyen} 5.37, $\Sigma_0^B$-definable]
\label{def:Sigma_0^B-definability}
A number (resp., string) function is $\Sigma_0^B$-definable \\
from a collection $\mathcal{L}$ of two-sorted functions and relations
if it is $p$-bounded and its (bit) graph is represented by a 
$\Sigma_0^B(\mathcal{L})$ formula.
\end{definition}

The notion of $\Sigma_0^B$-definability is different from
$\Sigma_0^B$-definability in a theory, which is concerned with
provability.   The two are related by the next result.

\begin{corollary}[\cite{CookNguyen} 5.38]\label{cor:definability}
Let $\mathcal{T} \supseteq V^0$ be a theory over $\mathcal{L}$
and assume that $\mathcal{T}$ proves the \\
$\Sigma_0^B(\mathcal{L})\comp$ axiom scheme.
  Then a function which is $\Sigma_0^B$-definable
from $\mathcal{L}$ is $\Sigma_0^B(\cal{L})$-definable in $\mathcal{T}$.
\end{corollary}

%% file: parityLformalizing.tex
\section{Formalizing $\#L$ and $\parityL$} 
\label{section:formalizing}

The first step is to prove $\# STCON$ and matrix powering over $\N$
are $AC^0$-complete for $\# L$.  The notion of an $AC^0$ reduction 
generalizes $\Sigma_0^B$-definability (Definition
\ref{def:Sigma_0^B-definability}).

\begin{definition}[\cite{CookNguyen} 9.1 $AC^0$-Reducibility]
\label{def:AC^0-reducibility}
A string function $F$ (respectively, a number function $f$) is 
$AC^0$-reducible to $\mathcal{L}$ if there is a sequence of string
functions $F_1$, \ldots, $F_n$, ($n\geq 0$) such that
$$ F_i \text{ is } \Sigma_0^B\text{-definable from } 
  \mathcal{L} \cup \{ F_1, \ldots, F_{i-1} \}
  \text{ for }
  i = 1, \ldots, n;
$$
and $F$ (resp. $f$) is $\Sigma_0^B$-definable from $\mathcal{L}
\cup \{ F_1, \ldots, F_{n}\}$.  A relation $R$ is $AC^0$-reducible to
$\mathcal{L}$ if there is a sequence of string functions 
$F_1$, \ldots, $F_n$ as above, and $R$ is represented by a 
$\Sigma_0^B(   \mathcal{L} \cup \{ F_1, \ldots, F_{n}\})$-formula.
\end{definition}

Notice that this is a semantic notion, separate from whether the
function $F$ (or $f$) is definable \emph{in a theory} (in the style of
Definition \ref{def:two-sorted-definability}).

\begin{definition}[\cite{CookNguyen} 5.15, Function Class]
\label{def:FC}
If $C$ is a two-sorted complexity class of relations, then the
corresponding function class $FC$ consists of all $p$-bounded number
functions whose graphs are in $C$, together with all $p$-bounded
string functions whose bit graphs are in $C$.
\end{definition}

This general definition yields the function classes $F\#L$ and
$F\parityL$ from the classes $\#L$ and $\parityL$.

\subsection{$\#L$  and its $AC^0$-complete problems}

\begin{definition} 
\label{def:numberL} 
  $\# L$ is the class of functions $F$
  such that there is a non-deterministic logspace Turing machine,
  halting in polynomial time on all inputs, which on input $X$, has
  exactly $F(X)$ accepting computation paths.\footnote{In general, we
  expect that the value of a function in $\#L$ is exponential;
  since we limit ourselves to 
  polynomially-bounded theories, the output of the function will be
  encoded as a binary string, and denoted with a capital letter:
  $F(X)$ instead of $f(X)$.} 
\end{definition}

The class $\#L$ is not known to be closed under $AC^0$-reductions, so we consider
instead its closure $AC^0(\#L)$, denoted $DET$ as in \cite{All04},
\cite{MV}, \cite{AllOg}, and others.

\begin{definition} \label{def:DET}
  $DET$ is the class of functions $F$ which are $AC^0$-reducible to
  $\#L$. 
\end{definition}

Since $DET$ is a function class, and we want a relation class, we
consider the class of relations with characteristic functions in
$DET$. 

The \emph{characteristic function} $f_R(\vec x, \vec X)$
of a relation $R(\vec x, \vec X)$ is defined as 
$$ f_R(\vec x, \vec X) = \left\{ \begin{array}{rl} 1 & \text{if }
  R(\vec x, \vec X) \\ 0 & \text{otherwise} \end{array} \right .$$
\begin{definition} 
  \label{def:RDET}
  $RDET$ is the class of relations whose characteristic functions are
  in $DET$.
\end{definition}



\begin{definition} \label{numberSTCON}
Let $\# STCON$ be the functional counting analog of the decision
problem $STCON$.  That is, $\# STCON$ is a function which,
given a directed graph\footnote{A simple graph.  Multiple edges are not
allowed.} with two distinguished nodes $s \neq t$ among its $n$
nodes, outputs the number of distinct paths of length $\leq p$ from
$s$ to $t$.  Let $G$ be the $n\times n$ Boolean matrix encoding the
adjacency matrix of the graph.
Then this number is represented (in binary) by the string
$\#STCON(n,s,t,p,G)$. 
\end{definition}

There are several simple ways to ensure that this output is finite.
It suffices to require that the the input graph be loop-free, but for
the reduction below it is more convenient to place an upper bound $p$
on the number of edges in an $s$-$t$ path.
We consider the output of $\#STCON$ as a number given in binary
notation (from least to most significant bit).  Thus $\#STCON$ is a
string function.

\begin{claim}
\label{claim:nSTCONcomplnL} 
$\# STCON$ is complete for $DET$ under $AC^0$ reductions. 
\end{claim}
\begin{proofof}{claim \ref{claim:nSTCONcomplnL}}
First, \underline{$\# STCON \text{ is } AC^0\text{-reducible to }
DET$}.  We show this by proving the stronger fact that $\#STCON \in
\# L$. 

The graph input to Turing machine $M$ is formatted specifically.  Let
the graph be represented by its Boolean adjacency matrix $G$, 
with $s$ and $t$ two listed vertices.  This matrix is
encoded as a binary string by means of the $\Row$ and pairing
functions.

$M$ maintains three numbers on its tape: the ``current'' vertex, the
``next'' vertex, and a count of the number of edges it has traversed.
These numbers are stored in binary.  The ``current'' vertex is
initialized to $s$
, and the count is initialized to $0$.

When run, $M$ traverses the graph by performing the following
algorithm.

\begin{algorithm}{Traverse}{n,s,t,p,G} \label{nSTCONalg}
  current \= s \\
  count \= 0 \\
  \begin{WHILE}{current \neq t \wedge count \leq p}
    next \= \text{nondeterministically-chosen number} < n \\
    \begin{IF}{G(current,next)} 
      current \= next
      \ELSE 
        \text{halt and reject}
    \end{IF} \\
    counter \= counter + 1
  \end{WHILE} \\
  \begin{IF}{current =t}
    \text{halt and accept}
    \ELSE
    \text{halt and reject}
  \end{IF}
\end{algorithm}

$M$ simulates a traversal of the graph from $s$ to $t$.  Every
accepting computation of $M$ traces a path from $s$ to $t$ (of length
$\leq p$), and for every path of length $\leq p$ from $s$ to $t$ there
is an accepting computation of $M$.  Thus $\# STCON \in DET$.

Next, \underline{$\# STCON$ is hard for $DET$} with respect to
$AC^0$-reducibility (Definition \ref{def:AC^0-reducibility}).  

Let $F_M(X) = $ the number of accepting computation paths of
nondeterministic logspace Turing machine $M$ on input $X$.


Below, we show that the string function $\#STCON$ is complete for
$DET$ under 
$AC^0$-reductions.  We $\Sigma_0^B$-bit-define an $AC^0$ function $H$
such that 
$$F_M(X) = \#STCON(n,s,t,p,H(X))$$
It follows that $\#STCON$
is $\Sigma_0^B$-definable from $DET$.

\underline{Defining $H$.}  $H$ can be defined using prior knowledge of
$M$.  As usual, logspace $M$ has two tapes, a read-only input
tape and a read-write work tape.  WLOG, let the work tape be infinite
in both directions, and initialized to all zeroes; also, let $M$ work
with the alphabet $\Sigma = \{0,1,\$ \}$,\footnote{$\$ $ is a special
symbol which used only to indicate the right end of the input
string.  It is safe to assume that $M$ keeps track of its input head
position and the left end of the input string, e.g., by maintaining a
counter of how many bits are to the left of the head.}  have a
``counter'' which increments with each step of computation, and have a
single accepting state.  $M$ runs in logarithmic space and
\emph{always} halts.  Thus, on input $X$, $M$ runs in time bounded
by $|X|^k + 1$,\footnote{The $+1$ here ensures that, even on the empty
input, $M$ can at least read its input.  If a tight time bound is
actually, e.g., $5|X|^j$, then $k=j+1$ will still yield $|X|^k +1$ a
bound.  For inputs $X$ short enough that $M$ cannot perform any
significant computation within the time bound, $M$ can read the input
and then look up the answers in a table.} where $k$ is some constant
specific to $M$.

\underline{Configurations of $M$.} Configurations of $M$ are
represented by $5$-tuples $(a,b,c,d,e)$ of numbers where
\begin{itemize}
  \item $a$ is the label of $M$'s current state, 
  \item $b$ is the position of the head on the input tape,
  \item $c$ is the value of the counter,
  \item $d$ is the numerical value of the binary contents the work tape
    to the left of the head, and 
  \item $e$ is the numerical value of the contents of the rest of the
    work tape, in reverse (so the least significant bit of $C$ is the
    bit currently being read).
\end{itemize}
$M$'s configuration represented by $(a,b,c,d,e)$ is encoded as the
number $\langle a,b,c,d,e \rangle$ using the pairing
function.  This encoding and the $left$ and $right$ projection
functions are used below in the $\Sigma_0^B$ formula representing the
bit-graph of $H$.  

\underline{$\Sigma_0^B$-defining $H$.}  
The string function $H(X)$ encodes the adjacency matrix of
nodes in the graph.  $H(X)(\ell,m)$ is true if there is a transition
of $M$ from the state encoded by $\ell$ to the state encoded by $m$.
\begin{align*}
H(X)( \ell, m) \leftrightarrow &
\big ( \exists a,b,c,d,e \leq (\ell+m), \psi \big) \vee \\  
& \exists a,b,c,d,e,a',b',c',d',e' \leq \ell,
\big( \ell = \langle a,b,c,d,e \rangle \wedge
m =  \langle a',b',c',d',e'\rangle \wedge \varphi
\big)
\end{align*}

\underline{Clauses of $\varphi$ correspond to transitions of $M$.}
The formula $\varphi$ is in disjunctive normal form.  For each possible
transition in $M$'s transition relation $\Delta$,
$\varphi$ has a clause specifying that there is an edge in the graph
between the nodes representing $M$'s configuration before and after
the transition.  For example, let $M$ be in state $\alpha_9$ reading a
$0$ on its input tape and $0$ on its work tape, with counter value
$17$ (below its cutoff limit).  Let one possible transition be to
write a $1$ to the work tape, move the work tape head right, move the
input tape head left, and update the counter.
$$
\begin{array}{c|ccc}
  & \text{before transition} & & \text{after transition} \\ 
  \hline
  \text{state:} & \alpha_9 & & \alpha_{12} \\
  \text{input tape:} &
  \begin{array}{ccccccccc} 
    & &  \bigtriangledown \\
    \cline{1-4} \cline{6-8}
    \multicolumn{1}{|c|}{1} &
    \multicolumn{1}{c|}{1} &
    \multicolumn{1}{c|}{0} &
    &
    \multicolumn{1}{c}{\cdots} &
    &
    \multicolumn{1}{|c|}{1} &
    \multicolumn{1}{c|}{\$} \\
    \cline{1-4} \cline{6-8}
  \end{array} & & 
  \begin{array}{ccccccccc} 
    &  \multicolumn{2}{c}{\bigtriangledown^\curvearrowleft} \\
    \cline{1-4} \cline{6-8}
    \multicolumn{1}{|c|}{1} &
    \multicolumn{1}{c|}{1} &
    \multicolumn{1}{c|}{0} &
    &
    \multicolumn{1}{c}{\cdots} &
    &
    \multicolumn{1}{|c|}{1} &
    \multicolumn{1}{c|}{\$} \\
    \cline{1-4} \cline{6-8}
  \end{array} \\ 
  \text{counter: } &  \fbox{17} &
  \xrightarrow{\text{transitions to}} & \fbox{18} \\
  \begin{array}{l} \\ 
    \text{ work tape: }\\
    \\ \\ 
  \end{array} &
  \begin{array}{ccccccccc} 
    & &  & 
      \bigtriangledown \\
    \hline
    \cdots &
    \multicolumn{1}{|c|}{0} &
    \multicolumn{1}{c|}{1} &
    \multicolumn{1}{c|}{0} &
    \multicolumn{1}{|c|}{0} &
    \multicolumn{1}{|c|}{1} &
    \cdots  \\
    \hline \\
    \\
  \end{array} & &
  \begin{array}{ccccccccc} 
    & & &  \multicolumn{2}{c}{ ^\curvearrowright \bigtriangledown} \\
    \hline
    \cdots &
    \multicolumn{1}{|c|}{0} &
    \multicolumn{1}{c|}{1} &
    \multicolumn{1}{c|}{1} &
    \multicolumn{1}{|c|}{0} &
    \multicolumn{1}{|c|}{1} &
    \cdots  \\
    \hline
    & & & \uparrow \\
    & & \multicolumn{4}{c}{\text{written bit}}
  \end{array}
\end{array}
$$
Then $\varphi$ contains the clause:
$$
(a=9 \wedge a'=12 \wedge \neg X(b) \wedge b+1 = b' \wedge c \leq
\underbrace{|X|\cdot|X|\cdots|X|}_{k \text{ times}} + 1 \wedge c'=c+1
  \wedge 2d+1 = d' \wedge 2e'=e) 
$$ 
Variables $d$, $d'$, $e$, and $e'$ store the numerical value of the
binary contents of the work tape.
$$
\begin{array}{c|} \hline \cdots \\ \hline \end{array}
\begin{array}{|c|} \hline 0 \\ \hline \end{array}
\begin{array}{|c|} \hline 1 \\ \hline \end{array}
\underbrace{
  \stackrel{\bigtriangledown}
	   {\begin{array}{|c|} \hline 0 \\ \hline \end{array}}
  \overbrace{
    \begin{array}{|c|} \hline 0 \\ \hline \end{array}
    \begin{array}{|c|} \hline 1 \\ \hline \end{array}
    \begin{array}{|c} \hline \cdots \hspace{1.2cm} \\ \hline \end{array}
  }^{e' \text{ in binary, reversed}}
}_{e \text{ in binary, reversed}}
\hspace{2mm}
\xrightarrow{\text{transitions to}}
\hspace{2mm}
\underbrace{
  \overbrace{
    \begin{array}{c|} \hline \cdots \\ \hline \end{array}
    \begin{array}{|c|} \hline 0 \\ \hline \end{array}
    \begin{array}{|c|} \hline 1 \\ \hline \end{array}
  }^{d \text{ in binary}}
  \begin{array}{|c|} \hline 1 \\ \hline \end{array}
}_{d' \text{in binary}}
\stackrel{^\curvearrowright \bigtriangledown \hspace{2mm}}
	 {\begin{array}{|c|} \hline 0 \\ \hline \end{array}}
\begin{array}{|c|} \hline 1 \\ \hline \end{array}
\begin{array}{|c} \hline \cdots \\ \hline \end{array}
$$ 

\underline{The right end of the input tape.}  In the special case when
$M$ reads $\$$ on its input tape, $b=|X|$.  For $b \geq |X|$, by
definition $X(b)$ is false; however, the intended interpretation is
that $X(b)$ is false iff the $b^\textrm{th}$ bit of $X$ is zero.  To
amend 
this, any clauses of $\varphi$ corresponding to transitions where $M$
reads $\$$ from the input tape will include the condition $b=|X|$
instead of $X(b)$ or $\neg X(b)$.  It is safe to assume that, having
read $\$$, $M$ \emph{never} moves further right on its input tape.

\underline{Specifying $s$ and $t$.}  The standard (above) for encoding
adjacency matrices specifies that the distinguished nodes $s$ and $t$
be associated with row/column $0$ and $1$, respectively.  Note that
$0$ and $1$ are not in the correct form of encoded $5$-tuples,
requiring special treatment in $\psi$, a disjunction of two clauses.
The first clause specifies that $H(X)(0,m)$ is true for $m$ encoding
$M$'s starting configuration.  Let $\alpha_0$ be $M$'s initial state:
$$ a=b=c=d=e=0 \wedge m = \langle a,b,c,d,e\rangle \wedge \ell=0 $$
The second clause specifies that $H(X)(\ell,1)$ is
true for any $\ell$ encoding a configuration where $M$ is in its
single accepting state $\alpha_r$:
$$ \ell = \langle a,b,c,d,e \rangle \wedge a = r \wedge 
c \leq \underbrace{|X|\cdot|X|\cdots|X|}_{k \text{ times}} + 1)
\wedge m=1 $$

Given $\varphi$ and $\psi$ as described above, we have a $\Sigma_0^B$
formula representing the bit-graph of $H$, a function which takes
input $X$ to logspace Turing machine $M$ and outputs the adjacency
matrix of $M$'s configurations (encoded as a binary string).  By
construction, the number $F_M(X)$ of accepting paths of machine $M$ on
input $X$ is exactly equal to the number of paths from $s$ (the node
numbered $0$) to $t$ (node $1$) in the adjacency matrix encoded in
$H(X)$.  By setting the bound $p$ on path length to be sufficiently
large, we see that $\# STCON$ is hard for $\# L$.

Since $\# STCON$ is both in $DET$ and hard for $DET$ with respect to
$AC^0$ reductions, it is complete for $DET$.
\end{proofof}

\begin{definition}[Matrix Powering] \label{def:matrix-powering}
  Given matrix $A$ and integer $k$, matrix powering is
  the problem of computing the matrix $A^k$.
\end{definition} 

As a function mapping $(A,k)\longmapsto A^k$, matrix powering does not
exactly fit the format required by Definition \ref{def:numberL}; it is
neither a number function nor a bit-graph.  Notationally, matrices are
represented as bit-strings; thus matrix powering is a string function.
In the context of $\# L$, we are interested in the bit-graph of this
string function.  Observe that the entries of the $k^\textrm{th}$
power of even a matrix over $\zo$ can have $k$ digits, so that the
$k^\textrm{th}$ power of a matrix must have entries encoded in binary.

Specifying matrices as strings requires an encoding
scheme for input matrix entries; \emph{unary matrix powering} and
\emph{binary matrix powering} are distinguished according to this
number encoding.  A unary input matrix is encoded by means of the
$\entry$ function; a binary input matrix is encoded by means of the
$\Rowtwo$ function (defined in Section \ref{subsection:str-mat-lists}).

\begin{remk} \label{remark:matrix-observation}
Let $A$ be the adjacency matrix of a graph; then $A[i,j]$
is $1$ if there is an edge from vertex $i$ to $j$, and $0$ otherwise.
Thus $A^1[i,j]$ is the number of edges from $i$ to $j$, that is, the
number of paths of length $1$ from $i$ to $j$.  $A^2[i,j] = \sum_{\ell
= 1}^n A[i,\ell] A[\ell,j]$ is the number of paths from $i$ to $j$
passing through one intermediate vertex, i.e., of length exactly $2$.
Inductively, $A^k[i,j]$ is the number of paths from $i$ to $j$ of
length exactly $k$.
\end{remk}

This observation provides the main insight for the next two lemmas.

\begin{lemma} \label{STCON<=MP}
$\# STCON$ is $AC^0$-reducible to matrix powering over $\N$.
\end{lemma}

\begin{proof}
Recall from Definition \ref{numberSTCON} that $\# STCON$ counts the
number of paths of length $\leq p$.  Matrix powering counts paths of
\emph{exactly} a given length.  It suffices to construct an $AC^0$
function to convert the adjacency matrix $G$ of a graph into the
adjacency matrix $G'$ of a different graph, such that every $s$-$t$
path of length $\leq p$ in $G$ is converted into an $s'$-$t'$ path of
length\footnote{The added constant is an artifact of the reduction.}
\emph{exactly} $p+2$ in the graph $G'$.  By notational convention, $G$
and $G'$ are be Boolean matrices; let $A'$ be the number matrix over
$\zo$ corresponding to $G'$, i.e., $A'[i,j]=1$ iff $G'(i,j)$.  Thus
$\#STCON(n,s,t,p,G) = A'^{(p+2)}[s',t']$.

The new graph $G'$ consists of $p+1$ ``layers'' of vertices; each
layer contains a copy of the vertices of $G$ (with no edges).  $G$
additionally has two vertices $s'$ and $t'$ not in any layer.  A pair
$(v,\ell)$ denotes the vertex $v$ in layer $\ell$; thus $(s,0)$ is the
copy of vertex $s$ in the first layer.  There is an edge from vertex
$(v,\ell)$ to $(v',\ell+1)$ if there is an edge from $v$ to $v'$ in
$G$.  There is an edge from $s'$ to $(s,0)$, an edge from each
$(t,\ell)$ to $t'$, and a self-loop on vertex $t'$.  Vertex $(x,y)$ in
$G'$ is identified with the number $\langle x, y\rangle$; also, $s'=0$
and $t'=1$.

Adopting the same convention as above, let the first two nodes listed
in $A$ be $s=0$ and $t=1$. Let $A$ be encoded as string $X$ and $A'$
be encoded as $X'=H(p,X)$.  The conversion can be $AC^0$ bit-defined
as follows:
\begin{eqnarray}
H(p,X)(i,j) &\leftrightarrow &
(i=0 \wedge j = \langle 0,0 \rangle) \vee \label{1} \\
&& (i=1 \wedge j=1) \vee \label{2} \\
&& \exists k \leq p (i=\langle 1,k\rangle \wedge j=1) \vee \label{3}
\\ 
&& \exists k < p \exists v,v' \leq |X| (X(v,v') \wedge i=\langle
v,k\rangle  
\wedge j = \langle v', k+1\rangle) \label{4}
\end{eqnarray}
The clause on line (\ref{1}) ensures that there is an edge from $s'$
to $(s,0)$.  Line (\ref{2}) defines the self-loop on vertex $t'$.
Line (\ref{3}) adds an edge from each $(t,\ell)$ to $t'$.  By the
condition on line (\ref{4}), the new graph appropriately ``inherits''
all edges from the original graph.

This construction ensures that, for every $s$-$t$ path in $G$ of
length $k$, there is a unique path in $G'$ from $(s,0)$ to $(t,k)$,
which can be uniquely extended to a path from $s'$ to $t'$ (passing
through $(s,0)$ and $(t,k)$) of length $k+2$.  Vertex $t'$ has only
one outgoing edge, a self-loop.  Thus for every $s$-$t$ path in $G$ of
length $\leq p$, there is a unique path in $G'$ from $(s,0)$ to $t'$
of length exactly $p+2$.

This completes the reduction, as 
$$\# STCON(n,s,t,p,G) = A'^{(p+2)}[s',t']$$
Thus $\# STCON$ is reducible to matrix powering.
Note that this reduction only requires powering a matrix with entries
in $\zo$.
\end{proof}

\begin{lemma} \label{lemma:MP<=nSTCON}
Unary matrix powering over $\N$ is $AC^0$-reducible to $\# STCON$.
\end{lemma}


\begin{proof}
Matrices above had entries from $\zo$. 
As a consequence, the $\Row$ and pairing functions sufficed to encode
them as strings.  However, when
each matrix entry requires more than one bit to store, another layer
of encoding is required.  We will use the $\entry$ function, with
definition given above by (\ref{equation:entry-def}).  This allows us
to encode a matrix of numbers $A = (x_{ij})$ as a string $X$,
recoverable as $x_{ij} = \entry(i,j,X) = (X^{[i]})^{j}$.

Since $A$ has entries from $\N$, it can be viewed as the adjacency
matrix of a multigraph.\footnote{Entries of $A$ that are $> 1$ can be
viewed as duplicate edges in the graph, e.g., $A[i,i]=4$ would mean
that node $i$ has four self-loops; $A[i,j]=3$ would mean that there
are three edges from node $i$ to node $j$.}  However, $\# STCON$ is
defined only on graphs with at most one edge between an ordered pair
of vertices.  A multigraph can easily be converted into a standard
graph by bisecting each edge with a new vertex.  If the former graph
$G$ had $E$ edges and $V$ vertices, then the new graph $G'$ will have
$2|E|$ edges and $|V|+|E|$ vertices.  The number of paths between any
two vertices in $G$ remains unchanged in $G'$, and the length of every
path is exactly doubled.

The following $AC^0$-defined string function $\Convert$ is designed to
perform this transformation from multigraph to graph.  Let $X$ be the
string encoding the unary $n \times n$ adjacency matrix $A$ of the
multigraph $G=(V,E)$, where $|V|=n$ and $|E|=\sum_i\sum_j A[i,j]$.  
Let $n' = n + |E|$. 
We construct the function $\Convert$ so that $\Convert(X) = A'$ is the
string encoding the $n'\times n'$ Boolean adjacency matrix
of the graph $G'$.
Let $\Convert$ be $AC^0$-bit-defined:
\begin{eqnarray*}
\Convert(X)(k,\ell) \leftrightarrow 
\exists i,j,c < |X| &\big(& 
(k=\langle i,j,c+n\rangle \wedge \ell = j \wedge c < \entry(i,j,X))\\ &&
\vee (\ell=\langle i,j,c+n\rangle \wedge k=i \wedge c < \entry(i,j,X))
\hspace{2mm} \big) 
\end{eqnarray*}
By construction, there are $A(i,j)$ vertices in $G'$ (with labels
$\langle i,j,0+n\rangle, \ldots, \langle i,j,c-1+n\rangle$) such that,
for each $c \in \{0,\ldots, c-1\}$, both 
$A'(i,\langle i,j,c+n\rangle)$ and $A'(\langle i,j,c+n\rangle,j)$ are
true.
(The added $n$ ensures that the vertex labels are unique, that is,
there is no new vertex labelled $\langle i,j,c+n\rangle = i'$ for some
$i'<n$ a label of a vertex in the original graph.)

Every pair of vertices in $G'$ shares at
most one edge; hence a single bit suffices to store each entry of
$A'$, and it is not necessary to use the $\seq$ function.
This definition of $\Convert(X)$ takes advantage of that
simplification by outputting a Boolean matrix.

By construction, for every $i,j \leq n$,
$$A^{k+1}[i,j] = A'^{2(k+1)}[i,j] 
= \# STCON(n',i,j,2(k+1),A')- \# STCON(n',i,j,2k,A')$$ 
The number $n'$ is not obviously computable, but is used here for
clarity.  It can be replaced with a larger number bounding $n'$ from
above, e.g., $|A|$.  This yields the same output of $\#STCON$, since
the string $\Convert(X)$ can be ``interpreted'' on larger numbers: any
vertex of index $\geq n$ has no in- or out-edges.

By remark \ref{remark:matrix-observation}, $A^{k+1}[i,j]$ is the
number of paths from 
$i$ to $j$ of length exactly $k+1$.  These paths map uniquely to paths
in $A'$ from $i$ to $j$ of length exactly $2(k+1)$.  (Notice that, by
construction, paths between vertices of the original graph will always
be of even length.) This is exactly the
number of paths from $i$ to $j$ of length $\leq 2(k+1)$ \emph{minus} the
number of paths from $i$ to $j$ of length $\leq 2k$.
Thus unary matrix powering is $AC^0$-reducible to $\# STCON$.
\end{proof}

\begin{claim} \label{claim:MPcomplnL}
 Unary matrix powering over $\N$ is $AC^0$-complete for $DET$.
\end{claim}

This follows from Lemma \ref{STCON<=MP}, Lemma \ref{lemma:MP<=nSTCON}, and
Claim \ref{claim:nSTCONcomplnL}.

\subsection{$\parityL$ and its $AC^0$-complete problems} 

\begin{definition} 
  \label{def:2L}
  $\parityL$ is the class of decision problems $R(\vec x, \vec X)$ such
  that, for some function $F \in \# L$, 
  $$ R(\vec x, \vec X) \leftrightarrow F(\vec x, \vec X)(0)$$
  That is, $R(\vec x, \vec X)$ holds depending on the value of $F(\vec
  x, \vec X) \bmod 2$.
\end{definition}

$\parityL$ is also called $ParityL$ or $MOD_2 L$.  In general, $MOD_k
L$ has a similar definition\footnote{That is, $MOD_kL$ is the class of
  decision problems $R(\vec x, \vec X)$ such that, for some function
  $F \in \#L$, $R(\vec x, \vec X)$ holds iff $F(\vec x, \vec X) \neq 0
  \bmod k$. 
  Using Fermat's little theorem, this can easily be modified to the
  two cases $F(x) =1 \bmod k$ and $F(x) = 0 \bmod k$ for any prime
  $k$.}. 

\begin{theorem}
\label{theorem:2L-closed}
  $\parityL$ is closed under $AC^0$-reductions.
\end{theorem}

Beigel et al. prove that $\parityL$ is closed under $NC^1$-reductions \cite{BGH},
which certainly implies it is closed under weaker
$AC^0$-reductions.
Chapter 9 of \cite{CookNguyen} provides a general criterion (Theorem 9.7) for such closure:
closure under finitely many applications of composition and string
comprehension.  \label{mod2lclosure}

By Definition \ref{def:matrix-powering}, matrix powering is a
function.  For convenience in aligning with Definition \ref{def:2L},
matrix powering mod 2 can also be considered as a decision problem on
tuples $(n,k,\langle i,j\rangle,A)$ where $A$ is a binary $n\times n$
matrix, and the answer is ``yes'' iff the $(i,j)^\textrm{th}$ entry of $A^k$
is $1$.
The proof of the following claim is nearly identical to that of Claim
\ref{claim:MPcomplnL}. 

\begin{claim} \label{claim:UMPcompleteMod2L}
Matrix powering over $\Z_2$ is $AC^0$-complete for $\parityL$.
\end{claim}

%% file: parityLmod_2L.tex
\section{A theory for $\parityL$} \label{section:2L} 

In this section, we develop a finitely axiomatized theory for the
complexity class $\parityL$.  Matrix powering is complete
for $\parityL$ under $AC^0$ reductions, and has the convenient
property that each matrix entry can be stored in a single bit.  

$V^0(2)$ provides the base theory for this section.  It is associated
with the complexity class $AC^0(2)$, which is $AC^0$ with the addition
of mod $2$ gates.  Section 9D of \cite{CookNguyen} develops this
theory, and the universal conservative extension $\overline{V^0(2)}$.

The theory $V\parityL$ is obtained from the base theory $V^0(2)$ by
adding an axiom which states the existence of a solution to matrix
powering over $\Z_2$.  $V\parityL$ is a theory over the language
$\mathcal{L}_A^2$; below, we detail two methods to obtain the added
$\Sigma_1^B(\mathcal{L}_A^2)$ axiom, either by defining it explicitly,
or by defining it in the universal conservative extension
$\overline{V^0(2)}$ and using results of \cite{CookNguyen} to find a
provably equivalent formula in the base language.

Throughout this section, matrices will be encoded using only the $\Row$
and pairing functions, taking advantage of the fact that a matrix
entry from $\zo$ can be stored by a single bit of a bit string.  For a
string $W$ encoding $n \times n$ matrix $M$, entries are recovered as
$M[i,j] = W(i,j)$.  (Or, more precisely, $M[i,j]=1$ iff $W(i,j)$.)

\subsection{The theory $V\parityL$} \label{subsection:V2L}

By the nature of its construction, the theory $V\parityL$ 
corresponds to $\parityL$.  
As we will prove, the set of provably total functions of $V\parityL$
exactly coincides 
with the functions of $F\parityL$ (Definition \ref{def:FC}) 
and the $\Delta_1^B$-definable  
relations of $V\parityL$ are exactly the relations in $\parityL$.

\begin{definition}[String Identity $ID(n)$] \label{def:ID(n)}
Let the $AC^0$ string function $ID(n)$ have output the string that
encodes the $n \times n$ identity matrix.  $ID(n)=Y$ has the 
$\Sigma_0^B(\mathcal{L}_{FAC^0})$ bit definition:
$$Y(b) \leftrightarrow 
\operatorname{\it left}(b) < n \wedge 
\Pair(b) \wedge 
\operatorname{\it left}(b)=\operatorname{\it right}(b) 
$$
\end{definition}

\begin{definition}[$\Powtwo$] \label{def:Pow}
Let $X$ be a string representing an $n\times n$ matrix over $\zo$.
Then the string function $\Powtwo(n,k,X)$ has output $X^k$, the string
representing the $k^\textrm{th}$ power of the same
matrix.\footnote{This section features a slight, but innocuous, abuse 
of notation: in mathematical expressions, string $X$ stands for the
matrix it represents, e.g., $X^3$ and $XY$.}  
\end{definition}

\begin{definition}[$\PowSeqtwo$] \label{def:PowSeq}
Let $X$ be a string representing an $n \times n$ matrix over $\zo$,
and let $X^i$ be the string representing the $i^\textrm{th}$ power of
the same matrix.  Then the string function $\PowSeqtwo(n,k,X)$ has output
the list of matrices $[ID(n), X,X^2,\ldots,X^k]$.
\end{definition}

Although Definition \ref{def:matrix-powering} specifies that matrix
powering is the problem of finding $X^k$, the function $\PowSeqtwo$
computes every entry of every power of $X$ up to the $k^\textrm{th}$
power. 
Lemmas \ref{lemma:Pow<=PowSeq} and \ref{lemma:PowSeq<=Pow} show that
$\Powtwo$ and $\PowSeqtwo$ are $AC^0$-reducible to each other.  It is more
convenient for us to develop the theory $V\parityL$ with an axiom
asserting the existence of a solution for $\PowSeqtwo$.

\begin{lemma} \label{lemma:Pow<=PowSeq}
$\Powtwo$ is $AC^0$-reducible to $\PowSeqtwo$.
\end{lemma}

$\Powtwo$ can be $\Sigma_0^B(\mathcal{L}_{FAC^0}\cup
\{\PowSeqtwo\})$-defined by:
\begin{equation} \label{eq:Pow<=PowSeq}
\Powtwo(n,k,X)(i) \leftrightarrow i<\langle n,n\rangle \wedge
\PowSeqtwo(n,k,X)^{[k]}(i) 
\end{equation}

\begin{lemma} \label{lemma:PowSeq<=Pow}
$\PowSeqtwo$ is $AC^0$-reducible to $\Powtwo$.
\end{lemma}

$\PowSeqtwo$ can be $\Sigma_0^B(\mathcal{L}_{FAC^0} \cup \{\Powtwo\})$-defined
by: 
\begin{equation}
  \PowSeqtwo(n,k,X)(i) \leftrightarrow
  i<\langle k,\langle n,n\rangle\rangle \wedge 
  \Pair(i) \wedge 
  \Powtwo(n, \operatorname{\it left}(i),X)(\operatorname{\it right}(i))
  \label{eq:PowSeq<=Pow}
\end{equation}

We detail two ways to define the relation
$\delta_{\PowSeqtwo}(n,k,X,Y)$ representing the graph of
$\PowSeqtwo(n,k,X)=Y$ 
in Sections \ref{subsection:implicit} and \ref{explicit}, below.

\begin{definition} 
The theory $V\parityL$ has vocabulary $\mathcal{L}_A^2$ and is
axiomatized by $V^0(2)$ and a $\Sigma_1^B(\mathcal{L}_A^2)$ axiom $PS_2$
(formula \ref{equation:implicit-PS}) 
stating the existence of a string value for the function
$\PowSeqtwo(n,k,X)$.
\end{definition}
The axiom is obtained below implicitly (formula 
(\ref{equation:implicit-PS}) in Section
\ref{subsection:implicit}) and explicitly (equation
(\ref{equation:explicit-PS}) in Section \ref{explicit}),
and is roughly equivalent to the statement ``there is some 
string $Z$ that witnesses the fact that $Y=\PowSeqtwo(n,k,X)$.''  
Thus it effectively 
states the existence of a solution to matrix powering over $\Z_2$.
Notice that it actually asserts the existence of the entire 
series of matrices $X^1, X^2, \ldots, X^k$, not just the matrix
$X^k$ as specified by Definition \ref{def:matrix-powering}. 

\subsection{Implicitly defining the new axiom}
\label{subsection:implicit} 
Using a series of intuitively ``helper'' functions, the
relation $\delta_{\PowSeqtwo}(n,k,X,Y)$ can be defined in the language
$\mathcal{L}_{FAC^0(2)} \supset \mathcal{L}_A^2$.  This method
requires the introduction of new function symbols, which can be used
to express the axiom $PS_2$ in $\overline{V^0(2)}$, a universal
conservative extension of $V^0(2)$.

Let $G(n,i,j,X_1,X_2)$ be the $AC^0$ string function which witnesses
the computation of the $(i,j)^\textrm{th}$ entry of the $n \times n$ matrix
product $X_1X_2$, bit-defined as:
$$ G(n,i,j,X_1,X_2)(b) \leftrightarrow b < n \wedge X_1(i,b) \wedge
X_2(b,j)$$
Each bit of $G(n,i,j,X_1,X_2)$ is the pairwise product of
the bits in the $i^\textrm{th}$ row of $X_1$ and the $j^\textrm{th}$ column of
$X_2$.  Thus the $(i,j)^\textrm{th}$ entry of the matrix product $X_1X_2
\bmod 2$ is 1 if and only if $PARITY(G(n,i,j,X_1,X_2))$ holds.
$\Prodtwo(n,X_1,X_2)$, the string function computing the product of two
matrices, can be bit-defined as:
$$ \Prodtwo(n,X_1,X_2)(i,j) \leftrightarrow i < n \wedge j<n \wedge
PARITY(G(n,i,j,X_1,X_2))$$ 
That is, the $(i,j)^\textrm{th}$ bit mod $2$ of the product matrix $X_1X_2$ is
$\sum_{b=0}^{n-1} X_1(i,b) X_2(b,j) \bmod 2$.  Each bit of the string 
$G(n,i,j,X_1,X_2)$ is one of the terms of this sum; hence
the parity of $G(n,i,j,X_1,X_2)$ is exactly the desired bit.

Observe that $G$ and $\Prodtwo$, as well as $\langle\cdot,\cdot\rangle$
and $\Pair$ (from Section \ref{notation}), have
$\Sigma_0^B(\mathcal{L}_{FAC^0(2)})$ 
definitions and are $AC^0(2)$ functions.  Let
$\delta_{\PowSeqtwo}(n,k,X,Y)$
be the $\Sigma_0^B(\mathcal{L}_{FAC^0(2)})$ formula
\begin{eqnarray}
  \forall b<|Y|,
  |Y|<\langle k,\langle n,n\rangle\rangle \wedge
  (Y(b) \supset \Pair(b)) 
  \wedge
  Y^{[0]} = ID(n) \wedge 
  \hspace{1in} \nonumber \\  \hspace{1in}
  \forall i<k ( Y^{[i+1]} = \Prodtwo(n,X,Y^{[i]})) 
  \label{formula:implicit-formula}
\end{eqnarray}

This formula asserts that the string $Y$ is the output of
$\PowSeqtwo(n,k,X)$.  The ``unimportant'' bits -- i.e., the bits of $Y$
that do not encode a piece of the list -- are required to be zero.
Thus $\PowSeqtwo(n,k,X)$ is the lexographically first string that encodes
the list of matrices $[X^1,X^2,\ldots,X^k]$.

$\overline{V^0(2)}$ is a conservative universal extension of $V^0(2)$,
defined in Section 9D of \cite{CookNguyen}. 
Theorem 9.67(b) of \cite{CookNguyen} asserts that there is a 
$\mathcal{L}_A^2$ term $t$ and a $\Sigma_0^B(\mathcal{L}_A^2)$ formula
$\alpha_{\PowSeqtwo}$ such that 
$$\exists Z<t,\alpha_{\PowSeqtwo}(n,k,X,Y,Z)$$
is provably equivalent to $(\ref{formula:implicit-formula})$ in
the theory $\overline{V^0(2)}$.

The axiom $PS_2$ used to define the theory $V\parityL$ is
\begin{equation} \label{equation:implicit-PS} 
\exists Y < m, \exists Z<t,\alpha_{\PowSeqtwo}(n,k,X,Y,Z)
\end{equation}

Formulas (\ref{formula:implicit-formula}) and
(\ref{equation:implicit-PS}) each specify that $Y$ witnesses the
intermediate strings $X^1$, $X^2$, \ldots, $X^k$.  String $Y$ is
\emph{not} required to witness any of the work performed in
calculating each entry of each product $X^j=X\times X^{j-1}$.  In
Section \ref{explicit}, $Y$ is used to witness \emph{all} of the
intermediate work, by using the notation $Y^{[x][y]}$ from Section
\ref{notation}.

\begin{lemma}\label{lemma:PowSeq_definable}
The matrix powering function $\PowSeqtwo$ is
$\Sigma_1^B(\mathcal{L}_A^2)$-definable in
$V\parityL$.
\end{lemma}
\begin{proof}
By construction,
$$ \PowSeqtwo(n,k,X) = Y \leftrightarrow \exists Z< t,
\alpha_{\PowSeqtwo}(n,k,X,Y,Z) $$
We need to show that
$$V\parityL \vdash \forall n,k \forall X \exists !  Y \exists Z<t,
\alpha_{\PowSeqtwo}(n,k,X,Y,Z)$$

The axiom $PS_2$ guarantees that such $Y$ and $Z$ exist.
(The function served by $Z$ is made explicit in Section \ref{explicit}.)
Formula (\ref{formula:implicit-formula}) uniquely specifies every bit of $Y$
using $n$, $k$, and $X$.
Consider the $\Sigma_0^B$-formula stating that the first $i$ bits of $Y$ are
unique.
By induction on this formula, $Y$ can be proved to be unique
 in the conservative extension $\overline{V^0(2)}$ together with the defining
axiom for $\PowSeqtwo$.
Thus $V\parityL$ also proves that $Y$ is unique.
\end{proof}

\begin{corollary} \label{cor:Powtwo_definable}
$\Powtwo$ is $\Sigma_1^B(\mathcal{L}_A^2)$-definable in $V\parityL$.
\end{corollary}
This corollary follows from Lemmas \ref{lemma:PowSeq_definable} and
\ref{lemma:Pow<=PowSeq} above.

In Lemma \ref{lemma:PowSeq*_definable} below, we show similarly that 
the aggregate function $\PowSeqtwostar$ is $\Sigma_1^B$-definable in
$V\parityL$.
Recall from Chapter 8 of \cite{CookNguyen} (Definition 8.9) that the aggregate
function $\PowSeqtwostar$ is the polynomially bounded string function that
satisfies
$$ | \PowSeqtwostar(b,W_1,W_2,X)| \leq \langle b, \langle |W_2|, \langle
|W_1|,|W_1|\rangle \rangle \rangle $$
and
$$ \PowSeqtwostar(b,W_1, W_2,X)(i,v) \leftrightarrow i<b \wedge
\PowSeqtwo((W_1)^{i},
(W_2)^{i},X^{[i]})(v)$$ 
The strings $W_1$, $W_2$, and $X$ encode
$b$-length lists of inputs to each place of $\PowSeqtwo$: $W_1$ encodes the
list of numbers $n_0,n_1,n_2,\ldots,n_{b-1}$; $W_2$ encodes the list
of numbers $k_0,k_1,\ldots,k_{b-1}$; and $X$ encodes the list of
strings $X_0,X_1, \ldots, X_{b-1}$; we are interested in raising the
$n_i\times n_i$ matrix encoded in $X_i$ to the powers
$1,2,\ldots,k_i$.  The string function $\PowSeqtwostar$ computes \emph{all}
these lists of powers; it aggregates many applications of the function
$\PowSeqtwo$.

\begin{lemma} \label{lemma:PowSeq*_definable}
The aggregate matrix powering function $\PowSeqtwostar$ is
$\Sigma_1^B(\mathcal{L}_A^2)$-definable in $V\parityL$.
\end{lemma}

\begin{proof}
  For $\delta_{\PowSeqtwostar}$ a $\Sigma_1^B$ formula, we need to show both:
  \begin{equation} \label{eq:PS*def_axiom}
 \PowSeqtwostar(b,W_1,W_2,X) 
  = Y \leftrightarrow \delta_{\PowSeqtwostar}(b,W_1,W_2,X,Y) 
  \end{equation}
  and
  $$V\parityL \vdash \forall b \forall X, W_1, W_2 \exists ! Y
  \delta_{\PowSeqtwostar}(b,W_1,W_2,X,Y)$$ 

  We introduce the $AC^0$ functions $\operatorname{\it max}$ and $S$ in order to
simplify the definition of
  $\delta_{\PowSeqtwostar}$ over $\overline{V^0(2)}$,
  and then use Theorem 9.67(b) of \cite{CookNguyen} to obtain a provably 
  equivalent $\Sigma_1^B(\mathcal{L}_A^2)$ formula.  
  Since $\overline{V^0(2)}$ is a conservative extension of $V^0(2)$, 
  this suffices to show that $\PowSeqtwostar$ is 
  $\Sigma_1^B(\mathcal{L}_A^2)$-definable in $V\parityL$.
  In order to do so, we must establish a $\Sigma_1^B(\mathcal{L}_{FAC^0(2)})$ 
  formula equivalent to the desired formula $\delta_{\PowSeqtwostar}$.

  Let the function $\operatorname{\it max}(n,W)$ yield the maximum number from a
list of $n$
  numbers encoded in string $W$:
  \begin{equation}
  \operatorname{\it max}(n,W)=x \leftrightarrow \exists i<n \forall j<n, x=
(W)^i \geq
  (W)^j \label{eq:max_def}
  \end{equation}

  The string function $S$ can be bit-defined as follows.
  Consider two strings $W_1$, representing a list of $b$ numbers, and
  $X$, representing a list of $b$ matrices as above.  The function
  $S(b,W_1,X)$ returns the matrix with matrices $X_i$ (appropriately
  padded with zeroes) on the diagonal, and all other entries zero.
  Let $m = \operatorname{\it max}(b,W_1)$.  Let $X_i'$ be the $m\times m$ matrix
$X_i$
  padded with columns and rows of zeroes:
  $$ \begin{array}{rc}
    X_i & \overbrace{\begin{array}{ccc} 0 & \ldots & 0 \end{array}}^{m-n_i} \\
    m-n_i \left \{ \begin{array}{c} 0 \\ \vdots \\ 0 \end{array} \right . &
    \begin{matrix} \ddots & & \vdots \\  &\ddots \\ \ldots && 0
    \end{matrix}
  \end{array} $$
  Then $S(b,Y_1,X)$ is the string encoding the matrix:
  $$ \begin{bmatrix}
    \begin{array}{cc}
      X_0' \\ & X_2'
    \end{array}
    & 0 \\
    0 &
    \begin{array}{cc}
      \ddots \\ & X_{b-1}'
    \end{array}
  \end{bmatrix} $$
  All entries not in the matrices along the diagonal are $0$.

  The string function $S$ can be bit-defined:
  \begin{eqnarray*}
    S(b,W_1,X)(i,j) &\leftrightarrow&
    \exists a<b, \exists i', j' < (W_1)^a, 
    i < m \wedge j< m \wedge \\
    &&     i=i' + m \wedge j=j'+m \wedge X^{[a]}(i',j')
  \end{eqnarray*}
  By convention, the unspecified bits (i.e., bits $b$ that are not
  pair numbers) are all zero.  This bit-definition ensures that the
  string $S(b,W_1,X))$ is uniquely defined.

  We will use this matrix $S(b,W_1,X)$ and the existence and uniqueness of its
sequence of matrix powers to show the existence and uniqueness of the aggregate
matrix powering function. 

  Let $n_{max}$ denote $max(b,W_1)$ and $k_{max}$ denote $max(b,W_2)$.
  Recall Definition \ref{formula:implicit-formula} of
$\delta_{\PowSeqtwo}(n,k,X)$.  
  By convention, the aggregate function of $\PowSeqtwo$ is defined as:
  \begin{eqnarray}
    \PowSeqtwostar(b,W_1,W_2,X) = Y  \leftrightarrow 
    |Y| < \langle b, \langle k_{max},\langle n_{max},n_{max}\rangle \rangle
\rangle
    \wedge 
    \hspace{1in} \nonumber \\
    \forall j < |Y|, \forall i < b,
    \big[ (Y(j) \supset \Pair(j)) \wedge 
    \delta_{\PowSeqtwo}((W_1)^i,(W_2)^i,X^{[i]},Y^{[i]}) \big] 
  \label{formula:implicit-PowSeq2*}
  \end{eqnarray}

  The right-hand side of (\ref{formula:implicit-PowSeq2*}) is the
  relation $\delta_{\PowSeqtwostar}(b,W_1,W_2,X,Y)$; it has a provably
equivalent 
  $\Sigma_1^B(\mathcal{L}_A^2)$-formula $\exists Z<t,
\alpha_{\PowSeqtwostar}(b,W_1,W_2,X,Y,Z)$,
  used as the definition for $\PowSeqtwostar$ over $V\parityL$.

  It remains to prove the existence and uniqueness for
$\PowSeqtwostar(b,W_1,W_2,X)$.
  The functions $\operatorname{\it max}$ and $S$ allow for a straightforward
proof based
 on the existence and uniqueness of the string $A=\PowSeqtwo(b\cdot n_{max},
k_{max}, S(b,W_1,X))$,
where $n_{max} = \operatorname{\it max}(b,W_1)$ and $k_{max} = \operatorname{\it
max}(b,W_2)$.

We would like to define the string $B = \PowSeqtwostar(b,W_1,W_2,X)$ from $A$.
Notice that $B$ encodes a list of strings, each of which represents a power of
the matrix $S(b,W_1,X)$.
The string $A$ encodes nearly the same information, but in a different format:
$A$ is a list of \emph{lists}, each of which encodes the powers of a matrix from
the list $X$ of matrices.

Observe that:
$$ S(b,W_1,X)^i = 
\begin{bmatrix}
    \begin{array}{cc}
      X_0' \\ & X_2'
    \end{array}
    & 0 \\
    0 &
    \begin{array}{cc}
      \ddots \\ & X_{b-1}'
    \end{array}
  \end{bmatrix}^i
= 
\begin{bmatrix}
    \begin{array}{cc}
      X_0'^i \\ & X_2'^i
    \end{array}
    & 0 \\
    0 &
    \begin{array}{cc}
      \ddots \\ & X_{b-1}'^i
    \end{array}
  \end{bmatrix} $$
Also, 
$$X_j'^i = 
\begin{array}{rc}
    X_j^i & \overbrace{\begin{array}{ccc} 0 & \ldots & 0 \end{array}}^{m-n_j} \\
    m-n_j \left \{ \begin{array}{c} 0 \\ \vdots \\ 0 \end{array} \right . &
    \begin{matrix} \ddots & & \vdots \\  &\ddots \\ \ldots && 0
    \end{matrix}
  \end{array} $$
Thus we can ``look up" the required powers of each matrix.
Let $A$ and $B$ be shorthand:
$$A=\PowSeqtwo( b\cdot n_{max}, k_{max}, S(b,W_1,X))$$
$$B = \PowSeqtwostar (b,W_1,W_2,X)$$
Then we can define $\PowSeqtwostar$ from $\PowSeqtwo$ as follows.
\begin{eqnarray}
{B^{[m][p]}}(i,j) &\leftrightarrow&
m < b \wedge p \leq (W_2)^m \wedge i < (W_1)^m \wedge j<(W_1)^m \wedge \nonumber
\\
&&\exists m'<m, m'+1=m \wedge
A^{[p]}(n_{max}\cdot m'+i, n_{max}\cdot m'+j)
\label{eq:PS2star-ABdef}
\end{eqnarray}
Here, $m$ represents the number of the matrix in the list $X$, $p$ represents
the power of matrix $X_m$, and $i$ and $j$ represent the row and column; thus
the formula above defines the bit $(X_m^p)(i,j)$ for all matrices $X_m$ in the
list $X$.
By shifting around the pieces of this formula and adding quantifiers for $m$,
$p$, $i$, and $j$, it is clear that (\ref{eq:PS2star-ABdef}) can be translated
into the appropriate form for a $\Sigma_1^B$-definition of the graph of
$\PowSeqtwostar$.

Thus existence and uniqueness of $\PowSeqtwostar(b,W_1,W_2,X)$ follow from
existence and uniqueness of $\PowSeqtwo(b\cdot n_{max},k_{max},S(b,W_1,X)$.
Since $\max$ and $S$ are $AC^0$ functions, they can be used without increasing
the complexity of the definition (by use of Theorem 9.67(b) of
\cite{CookNguyen}, as above).
\end{proof}

\subsection{Explicitly defining the new axiom} \label{explicit}

The application of Theorem 9.67(b) in the previous section allowed for
an easy conversion from a readable formula (\ref{formula:implicit-formula})
expressing the graph $\delta_{\PowSeqtwo}(n,k,X,Y)$ to a
$\Sigma_1^B(\mathcal{L}_A^2)$ formula (\ref{equation:implicit-PS}) for $PS_2$.
In this section, we explicitly write a
$\Sigma_1^B(\mathcal{L}_{A}^2(PARITY))$ formula for $PS_2$ without
helper functions.  This formula clarifies the function of string $Z$
in formula (\ref{equation:implicit-PS}).

We can write an explicit $\Sigma_1^B(\mathcal{L}_{FAC^0(2)})$ formula
for $PS_2$ using $Z$ to witness \emph{all} of the intermediate work done
in computing each bit of each power of $X$: $X^1, X^2, \ldots, X^k$.
Recall from Section \ref{notation} the notation $Z^{[i]}$ for the
string in the $i^\textrm{th}$ row of the list of strings encoded in $Z$, and 
$Z^{[i][j]}$ for the $(i,j)^\textrm{th}$ string entry of the two-dimensional
matrix of strings encoded in string $Z$.  In particular, the (very
long!) string $Z$ encodes a list of $k$ strings $Z^{[\ell]}$ such that
for for $\ell < k$,
$Z^{[\ell][\langle i,j\rangle]}$ is the string witnessing the
computation of the $(i,j)^\textrm{th}$ bit of $X^{\ell+1}$ for $i,j<n$.
\label{Row_2_usage}

With this organization of $Z$, $\delta_{\PowSeqtwo}(n,k,X,Y)$ can be
$\Sigma_1^B(\mathcal{L}_{FAC^0(2)})$ defined as:
\begin{eqnarray} 
\exists Z < \langle k,\langle n,n\rangle\rangle, 
\big(Y^{[0]} = ID(n)\big) \wedge 
\Big( \forall \ell<k, \forall i,j<n,
 \big( 
Z^{[\ell+1][\langle i,j\rangle]}(b) \leftrightarrow 
X(i,b) \wedge Y^{[\ell]}(b,j) \big) 
\nonumber \\ 
\wedge Y^{[\ell+1]}(i,j) \leftrightarrow
PARITY(Z^{[\ell+1][\langle i,j\rangle]})
\Big) \hspace{5mm}
\label{explicit-formula}
\end{eqnarray}
For all $0 \leq \ell \leq k$, the string $Y^{[\ell]}$ encodes the
$\ell^\textrm{th}$ power of matrix $X$.
The sequence of matrices $X^0$, $X^1$, $X^2$, \ldots, $X^k$ are encoded as
strings $Y^{[0]}, Y^{[1]}, \ldots, Y^{[k]}$. 

We require a $\Sigma_1^B(\mathcal{L}_A^2)$ axiom for $\exists Y<m, 
\delta_{\PowSeqtwo}(n,k,X,Y)$.  Unfortunately, formula
(\ref{explicit-formula}) above includes $PARITY$, which is not a
function in $\mathcal{L}_A^2$.  The usage of $PARITY$ is most
conveniently eliminated by application of Theorem 9.67(b), as in the
previous section, obtaining a $\Sigma_1^B(\mathcal{L}_A^2)$ formula
$\exists Z<t, \beta_{\PowSeqtwo}(n,k,X,Y,Z)$ which $\overline{V^0(2)}$
proves equivalent to 
(\ref{explicit-formula}).  This formula can be made
explicit by the construction in the proof of the First Elimination
Theorem (9.17), which forms the basis of Theorem 9.67(b).  This
construction provides no special insight for the present case,
however, and so we omit it here.  

The $\Sigma_1^B(\mathcal{L}_A^2)$ axiom $PS_2$ for the theory
$V\parityL$ is: 
\begin{equation} \label{equation:explicit-PS}
 \exists Y<m, \exists Z<t, \beta_{\PowSeqtwo}(n,k,X,Y,Z)
\end{equation}
Notice that this appears identical to the axiom
(\ref{equation:implicit-PS}) 
obtained in Section \ref{subsection:implicit} by less explicit 
means.  The differences between (\ref{equation:implicit-PS}) and
(\ref{equation:explicit-PS}) are obscured in the formulae 
$\alpha_{\PowSeqtwo}$ and $\beta_{\PowSeqtwo}$.

\subsection{The theory $\overline{V\parityL}$}
\label{subsection:overline-V2L}

Here we develop the theory $\overline{V\parityL}$, a universal
conservative extension of $V\parityL$.  Its language
$\mathcal{L}_{F\parityL}$ contains function symbols for all 
string functions in $F\parityL$.  The defining axioms for the 
functions in $\mathcal{L}_{F\parityL}$ are based on their $AC^0$ 
reductions to the matrix powering function.  
Additionally, $\overline{V\parityL}$ has a
quantifier-free defining axiom for $\PowSeqtwo'$, a string function 
with inputs and output the same as $\PowSeqtwo$.

Formally, we leave the definition of
$\PowSeqtwo$ unchanged; let $\PowSeqtwo'$ be the function with the same
value as $\PowSeqtwo$ but the following quantifier-free defining axiom:
\begin{eqnarray}
  |Y|<\langle k,\langle n,n\rangle\rangle \wedge
  \big( b<|Y| \supset
  (\neg Y(b) \vee \Pair(b)) \big)
  \wedge
  Y^{[1]} = X \wedge 
   \hspace{1in} \nonumber \\  \hspace{2in}
  \big(
  i<k\supset  Y^{[i+1]} = \Prodtwo(n,X,Y^{[i]})
  \big)
  \label{formula:quant-free-def}
\end{eqnarray}
Notice that this formula is similar to
(\ref{formula:implicit-formula}), but has new free variables $b$ and
$i$.  The function $\PowSeqtwo$ satisfies this defining axiom for
$\PowSeqtwo'$.  $V^0(2)$, together with both defining axioms, proves
$\PowSeqtwo(n,k,X)=\PowSeqtwo'(n,k,X)$.

We use the following notation.  For a given formula $\varphi(z,\vec x,
\vec X)$ and $\mathcal{L}_{A}^2$-term $t(\vec x, \vec X)$, let
$F_{\varphi,t}(\vec x, \vec X)$ be the string function with bit
definition
\begin{equation} \label{formula:FvarphiDefAxiom}
  Y = F_{\varphi(z),t}(\vec x, \vec X) \leftrightarrow
  z < t(\vec x, \vec X) \wedge \varphi(z,\vec x, \vec X)
\end{equation}

\begin{definition}[$\mathcal{L}_{F\parityL}$]
We inductively define the language $\mathcal{L}_{F\parityL}$ of all
functions with (bit) graphs in $\parityL$.
Let $\mathcal{L}_{F\parityL}^0 = \mathcal{L}_{FAC^0(2)}\cup
\{\PowSeqtwo'\}$.  
Let $\varphi(z,\vec x,\vec X)$ be an open formula over
$\mathcal{L}_{F\parityL}^i$, and let $t=t(\vec x, \vec X)$ be a
$\mathcal{L}_A^2$-term.  
Then $\mathcal{L}_{F\parityL}^{i+1}$ is $\mathcal{L}_{F\parityL}^i$
together with the string function $F_{\varphi(z),t}$ that has defining 
axiom:
\begin{equation} \label{equation:F_varphi}
F_{\varphi(z),t}(\vec x, \vec X)(z) \leftrightarrow z < t(\vec x, \vec
X) \wedge \varphi(z,\vec x, \vec X)
\end{equation}

Let $\mathcal{L}_{F\parityL}$ be the union of the languages
$\mathcal{L}_{F\parityL}^i$.  Thus $\mathcal{L}_{F\parityL}$ is the
smallest set containing $\mathcal{L}_{FAC^0(2)}\cup\{\PowSeqtwo'\}$ and
with the defining axioms for the functions $F_{\varphi(z),t}$ for
every open $\mathcal{L}_{F\parityL}$ formula $\varphi(z)$.
\end{definition}

Notice that $\mathcal{L}_{F\parityL}$ has a symbol for every string
function in $F\parityL$.  By Exercise 9.2 of \cite{CookNguyen}, for every number
function $f$
in $F\parityL$, there is a string function $F$ in $F\parityL$ such that
$f=|F|$.  Thus there is a term in $\mathcal{L}_{F\parityL}$ for every
number function in $F\parityL$.

\begin{definition}[$\overline{V\parityL}$] The universal theory
$\overline{V\parityL}$ over language $\mathcal{L}_{F\parityL}$ has the
axioms of  
$\overline{V^0(2)}$ together with the quantifier-free defining axiom 
(\ref{formula:quant-free-def}) for $\PowSeqtwo'$ and the above defining
axioms (\ref{equation:F_varphi}) for the functions
$F_{\varphi(z),t}$. 
\end{definition}

\begin{theorem} \label{theorem:overline-UCE-2L}
$\overline{V\parityL}$ is a universal conservative extension of
  $V\parityL$.
\end{theorem}

\begin{proof}
Obviously, $\overline{V\parityL}$ is a universal theory.  It is an
extension of $V\parityL$ because all axioms of $V\parityL$ are
theorems of $\overline{V\parityL}$.

To show that it is conservative, we build a sequence of theories
$\mathcal{T}_i$ 
analogous to the languages $\mathcal{L}_{F\parityL}^i$, and show that
each $\mathcal{T}_{i+1}$ is a universal conservative extension of
$\mathcal{T}_i$.  By Theorem 
5.27 in \cite{CookNguyen} (extension by definition), it
suffices to show that $F_{i+1}$ is definable in $\mathcal{T}_i$.
(Alternatively, (\ref{equation:F_varphi}) provides a bit-definition of
$F_{i+1}$, so Corollary 5.39 of \cite{CookNguyen} gives the same result.)

Let $\mathcal{T}_0 = V\parityL$, and let $\mathcal{T}_{i+1}$ be the
language obtained from $T_i$ by adding a new function $F_{i+1}$ of the
form $F_{\varphi(z),t}$ and its defining axiom (\ref{equation:F_varphi}), where
$\varphi(z)$ is a quantifier-free formula in the language
$\mathcal{L}^i$ of $\mathcal{T}^i$.  Thus
$\overline{V\parityL}$ extends every theory $\mathcal{T}^i$; it is their
union.
$$ \overline{V\parityL} = \bigcup_{i \geq 0} \mathcal{T}^i $$

The rest of the proof consists of proving the claim that, for each
$i\geq 0$, the function $F_{i+1}$ is definable in $\mathcal{T}^i$.
We proceed by induction, using results from \cite{CookNguyen}.

As a base case, $V\parityL$ proves $\Sigma_0^B(\mathcal{L}_{FAC^0(2)})
\comp$ by Lemma 9.67(c), since it extends $V^0(2)$.  Formula
(\ref{formula:implicit-formula}) gives a
$\Sigma_0^B(\mathcal{L}_{FAC^0(2)})$ defining axiom for $\PowSeqtwo$ (and
$\PowSeqtwo'$).  By Lemma 9.22, 
$V\parityL(\PowSeqtwo, \PowSeqtwostar)$ proves (\ref{equation:165}) for
$\PowSeqtwo$;
also,  both $\PowSeqtwo$ and $\PowSeqtwostar$ are 
$\Sigma_0^B(\mathcal{L}_{FAC^0(2)})$-definable in $V\parityL$.
(Notice that formulas (\ref{formula:implicit-formula}) and
(\ref{formula:implicit-PowSeq2*}) provide exactly 
these definitions for $\PowSeqtwo$ and $\PowSeqtwostar$, respectively.) 
By Lemma 9.22 and Theorem 8.15, $V\parityL$ proves 
$\Sigma_0^B(\mathcal{L}_{FAC^0(2)} \cup \{ \PowSeqtwo \}) \comp$.

Inductively, $\mathcal{T}^i$ proves $\Sigma_0^B(\mathcal{L}^{i-1})
\comp$, $F_i$ and $F_i^*$ are
$\Sigma_0^B(\mathcal{L}^{i-1})$-definable in $\mathcal{T}^i$, and
$\mathcal{T}^i(F_i,F_i^*)$ proves equation (165):
\begin{equation} \label{equation:165}
 \forall i<b, F_i^*(b,\vec Z, \vec X)^{[i]} =
F_i((Z_1)^i, \ldots, (Z_k)^i, X_1^{[i]}, \ldots, X_n^{[i]})
\end{equation}

Thus by Theorem 8.15, $\mathcal{T}^i$ proves
$\Sigma_0^B(\mathcal{L}^i) \comp$. 

By construction, $F_{i+1}(\vec x, \vec X)$ has defining axiom
(\ref{formula:FvarphiDefAxiom}) for $\varphi$ some open
$\mathcal{L}^i$ formula.  Thus by Lemma 9.22, $F_{i+1}$ is definable
in $\mathcal{T}^i$.  (Also, the next inductive hypothesis established:
$\mathcal{T}^{i+1}$ extends $\mathcal{T}^i$, which can
$\Sigma_0^B(\mathcal{L}^i)$-define $F_i$ and $F_i^*$ and prove
(\ref{equation:165}) for $F_i$ and $F_i^*$.)
\end{proof}

\subsection{Provably total functions of $V\parityL$}
\label{section:provably-total-V2L}

The remaining results follow directly from \cite{CookNguyen}.  We
restate them here for convenience.

\begin{claim}
The theory $\overline{V\parityL}$ proves the axiom schemes
$\Sigma_0^B(\mathcal{L}_{FAC^0(2)}) \comp$, \\
$\Sigma_0^B(\mathcal{L}_{FAC^0(2)}) \ind$, and
$\Sigma_0^B(\mathcal{L}_{FAC^0(2)}) \minaxiom$.
\end{claim}
See Lemma 9.23 \cite{CookNguyen}.

\begin{claim}
\begin{enumerate}
  \renewcommand{\labelenumi}{(\alph{enumi})}
  \item A string function is in $F\parityL$ if and only if it is represented
    by a string function symbol in $\mathcal{L}_{F\parityL}$.
   \item A relation is in $\parityL$ if and only if it is represented by an open
formula of $\mathcal{L}_{F\parityL}$ if and only if it is
    represented by a $\Sigma_0^B(\mathcal{L}_{F\parityL})$ formula.
 \end{enumerate}
\end{claim}
See Lemma 9.24 \cite{CookNguyen}.

\begin{corollary}
Every $\Sigma_1^B(\mathcal{L}_{F\parityL})$ formula $\varphi^+$ is
equivalent in $\overline{V\parityL}$ to a $\Sigma_1^B(\mathcal{L}_A^2)$
formula $\varphi$.
\end{corollary}
See Corollary 9.25 \cite{CookNguyen}.

\begin{corollary} \label{cor:definability-2L}
\begin{enumerate}
    \renewcommand{\labelenumi}{(\alph{enumi})}
    \item A function is in $F\parityL$ iff it is
      $\Sigma_1^B(\mathcal{L}_{F\parityL})$-definable in
      $\overline{V\parityL}$ 
      iff it is $\Sigma_1^B(\mathcal{L}_{A}^2)$-definable in
      $\overline{V\parityL}$.
    \item A relation is in $\parityL$ iff it is
      $\Delta_1^B(\mathcal{L}_{F\parityL})$-definable in
      $\overline{V\parityL}$ iff it is
      $\Delta_1^B(\mathcal{L}_A^2)$-definable in $\overline{V\parityL}$.
\end{enumerate}
\end{corollary}
See Corollary 9.26 \cite{CookNguyen}.

The next two theorems follow from Corollary \ref{cor:definability}
and Theorem \ref{theorem:overline-UCE-2L}.
Their analogs are Theorem 9.10 and Corollary 9.11 \cite{CookNguyen}.
\begin{theorem} \label{theorem:provablytotalF2L}
A function is provably total ($\Sigma_1^1$-definable) in $V\parityL$ iff
it is in $F\parityL$.
\end{theorem}

\begin{theorem}
  A relation is in $\parityL$ iff it is
  $\Delta_1^B$-definable in $V\parityL$ iff it is $\Delta_1^1$-definable
  in $V\parityL$.
\end{theorem}


%% file: parityLnumberL.tex
\section{A theory for $\#L$}
\label{section:DET}

In this section, we develop a finitely axiomatized theory for the
complexity class $DET$, the $AC^0$-closure of $\#L$.  Claims
\ref{claim:nSTCONcomplnL} and \ref{claim:MPcomplnL} showed that
$\#STCON$ and unary matrix powering over $\N$ are complete for $DET$.

By adopting the notation $DET$, we have invoked the fact that
computing the determinant of an integer-valued matrix is
$AC^0$-complete for $\# L$.  However, the established notation has
only two types: 
numbers $\in \N$ and binary strings.  In this section it will be
convenient to be able to capture negative numbers, so we introduce a
method for encoding them as strings.  This binary encoding will also
be convenient in 
computing the determinant, which we expect to be a large value.

$VTC^0$ provides the base theory for this section.  It is associated
with the complexity class $TC^0$, which is $AC^0$ with the addition of
threshold gates.  Section 9C of \cite{CookNguyen} develops this
theory, and its universal conservative extension $\overline{VTC^0}$.
There, it is shown that the function $\numones$ is $AC^0$-complete for 
$TC^0$ (where $numones(y,X)$ is the number of elements of $X$ that
are $< y$, i.e., the number of $1$ bits of $X$ among its first $y$
bits.)  Thus $\mathcal{L}_{VTC^0}$ is $\mathcal{L}_A^2$.
$VTC^0$ extends $V^0$ by the 
addition of an axiom for $numones$.

The theory $V\# L$ is obtained from the theory $VTC^0$ by the
addition of an axiom which states the existence of a solution to
matrix powering over $\Z$.  It is a theory over the base language
$\mathcal{L}_A^2$.  The added $\Sigma_1^B(\mathcal{L}_A^2)$ axiom is
obtained below by a method similar to the previous section. 
The universal conservative extension $\overline{V\#L}$ is obtained as
before, and its language $\mathcal{L}_{F\#L}$ contains symbols for all
string functions of $F\#L$.

\subsection{Encoding integers in bit-strings}

We encode an integer $x \in \Z$ as a bit-string $X$ in the following
way.  The first bit $X(0)$ indicates the ``negativeness'' (sign) of
$x$: $x< 0$ iff $X(0)$.
The rest of $X$ consists of a binary representation of $x$, from least 
to most significant bit.
This section defines addition and multiplication for integers encoded
in this way, and extends the similar functions given in
\cite{CookNguyen} for binary encodings of natural numbers.

\cite{CookNguyen} includes notation for encoding a number $n\in \N$ as
a binary string $X$.
$$ \bin(X) = \sum_i X(i) \cdot 2^i $$
Addditionally, for $X$ and $Y$ two strings, the string functions
``binary addition'' 
$X+Y$ and ``binary multiplication'' $X \times Y$ are defined (Sections
4C.2 and 5.2).

Our string $X$ is shifted one bit to accomodate the sign of integer
$x$. 
We define the number function $\intsize$ analogously, to map from
a binary string $X$ (representing an integer) to its size $|X|$.  Thus
the integer $x$ can be
recovered\footnote{This is a slight abuse of notation, since $X(i)$ is
  true/false, not one/zero-valued.} 
 as $(-1)^{X(0)} \cdot \intsize(X)$.
$$ \intsize (X) = \sum_i X(i+1) \cdot 2^i $$

We write $X+_\Z Y=Z$ for the string function ``integer addition'' and
$X\times_\Z Y=Z$ for the string function ``integer multiplication.''
Define the relations $R_{+_\Z}$ and $R_{\times_Z}$ by
$$ R_{+_\Z}(X,Y,Z) \leftrightarrow
(-1)^{Z(0)} \cdot \intsize (Z) = (-1)^{X(0)} \cdot \intsize(X) +
(-1)^{Y(0)} \cdot \intsize(Y) $$
$$ R_{\times_\Z}(X,Y,Z) \leftrightarrow
\intsize(Z) = \intsize(X) \times \intsize(Y) 
\wedge 
 ( Z(0) \leftrightarrow (X(0) \oplus Y(0))) 
$$
Here, $\oplus$ represents exclusive or.

The bit-definitions for these functions will be similar to the
definitions for binary numbers in \cite{CookNguyen} (binary addition
in Chapter 4, binary multiplication in Chapter 9), and additionally
handle the complication of having signed numbers.

We adopt the relation $\Carry$ from \cite{CookNguyen}, and modify it
to work with encoded integers.  $\Carry_\Z(i,X,Y)$ holds iff both
integers represented by $X$ and $Y$ have the same sign, and there is
a carry into bit $i$ when computing $X+Y$.
$$ \begin{array}{rl}
\Carry_\Z(i,X,Y) \leftrightarrow &
(X(0) \leftrightarrow Y(0)) \wedge \\
& \exists k<i, k>0 \wedge X(k) \wedge Y(k) \wedge
\forall j<i [k<j \supset (X(j) \vee Y(j))] 
\end{array} $$
Notice that $\Carry$ includes the check that $X$ and $Y$ have the same
sign. 

When the integers have different signs, we need to perform
subtraction.  The order of subtraction will depend on which of the two
integers $X$ and $Y$ has larger size.
\newcommand{\FirstIntDominates}{\operatorname{\it FirstIntDominates}}
The relation $\FirstIntDominates(X,Y)$ holds iff
$\intsize(X) > \intsize(Y)$. 
$$ \begin{array}{rl}
\FirstIntDominates(X,Y) \leftrightarrow
|X| \geq |Y| \wedge & 
\exists k \leq |X|, (X(k) \wedge \neg Y(k)  \wedge \\
& (\forall j \leq |X|,  (k<j \wedge Y(j)) \supset X(j))
\end{array} $$
Suppose that $\intsize(X) > \intsize(Y)$ and the integers have
different signs.  Then we can think of the subtraction 
$\intsize(X) - \intsize(Y) = \intsize(Z)$ as computing, bit by bit,
the integer $Z$ which adds to $Y$ to obtain $X$.  The relation
$\Borrow(i,X,Y)$ holds 
iff $\intsize{X} > \intsize{Y}$ and there is a carry from bit $i$
when performing the addition $Z+Y$ (that is, the $i^\textrm{th}$ bit
of $X$ is ``borrowed from'' in the subtraction).
$$ \begin{array}{rl}
\Borrow(i,X,Y) \leftrightarrow &
(X(0) \leftrightarrow \neg Y(0)) \FirstIntDominates(X,Y) \wedge \\
& \exists k<i, k>0 \wedge \neg X(k) \wedge Y(k) \wedge
\forall j<i (k<j \supset (\neg X(j) \vee Y(j)))
\end{array} $$
Notice that $\Borrow$ includes the check that $X$ and $Y$ have
different signs, and $\intsize(X) > \intsize(Y)$.

Given these relations, it is now possible to bit-define integer
addition:
\begin{eqnarray}
(X +_\Z Y)(i) & \leftrightarrow &
\big[
\big( i=0 \wedge X(0) \wedge \FirstIntDominates(X,Y) \big)
\vee 
  \label{eq:sign-right1} \\
&& \big(i=0 \wedge Y(0) \wedge \FirstIntDominates(Y,X) \big)
\big]                                 
  \label{eq:sign-right}  \\ 
&& \vee \big[ 
    i>0 \wedge X(i) \oplus Y(i) \oplus \\
&& \big(\Carry_\Z(i,X,Y)
  \label{eq:carry-right} 
  \vee \Borrow(i,X,Y)
  \vee \Borrow(i,Y,X)
  \big) \big] 
  \hspace{3mm}
\end{eqnarray}
Lines \ref{eq:sign-right1} and \ref{eq:sign-right} ensure that the
sign of the resulting integer is correct.  The clauses on line
\ref{eq:carry-right} are mutually exclusive; at most one of them can
be true.  Notice that each of these clauses applies to a particular
case: 
\begin{itemize}
  \item $\Carry_\Z(i,X,Y)$ applies when $X$ and $Y$ have the same
    sign;
  \item $\Borrow(i,X,Y)$ applies when 
    $\intsize(X) > \intsize(Y)$
    and 
    $X$ and $Y$ have different signs; and 
  \item $\Borrow(i,Y,X)$ applies when 
    $\intsize(Y)>\intsize(X)$
    and 
    $X$ and $Y$ have different signs.   
\end{itemize}
In the special case when $\intsize(X)=\intsize(Y)$ and $X$ and $Y$
have different signs, neither of the
$\Borrow$ clauses will apply.  This has the desired effect:
all bits of $X+_\Z Y$ will be zero.  (Notice that there is only one
valid encoding of zero, as the all-zero string $+0$.)

Binary multiplication is $\Sigma_1^B$-definable in $VTC^0$ by results
from Section 9C.6 of \cite{CookNguyen}.  
It is easily adaptable to integer multiplication.
Define the string function 
\newcommand{\BinaryPart}{\operatorname{\it BinaryPart}}
$\BinaryPart$ such that, for an integer encoded as string $X$,
$\bin(\BinaryPart(X)) = \intsize(X)$.
$$ \BinaryPart(X)(i) \leftrightarrow X(i+1)$$
This function simply extracts the part of the string encoding the
number, and allows for an easy definition of integer multiplication.
\begin{eqnarray}
(X \times_\Z Y)(i) & \leftrightarrow &
\big(i=0 \wedge (X(0) \leftrightarrow \neg Y(0))\big) \vee \nonumber \\
&& \big( i>0 \wedge \exists k <i,
 k+1=i \wedge (\BinaryPart(X) \times \BinaryPart(Y))(k) \big)
\hspace{10mm}
\label{eq:times_Z}
\end{eqnarray}

The formulas given by (\ref{eq:sign-right1} -- \ref{eq:carry-right})
and (\ref{eq:times_Z}) define integer addition $+_\Z$ and integer
multiplication $\times_\Z$.  In order to use them in the next section,
we will work in $\overline{VTC^0}$.


\subsection{Additional complete problems for $\#L$} 
\label{section:furthernLproblems}

Eventually, we would like to justify the fact that $AC^0(\#L)= DET$ by
formalizing the computation of integer determinants in $V\#L$.  To
that end, the theory $V\#L$ will be formed from the base theory
$VTC^0$ by the addition of an axiom stating the existence of a
solution for integer matrix powering.
Computation of integer determinants is reducible to matrix powering
\cite{Berkowitz}. 

The choice of matrix powering for our axiom is further justified by the
fact that integer matrix powering is $\in \#L$ (as shown by
Claim \ref{claim:ZMPcomplnL} below) and the easy reduction from matrix
powering over $\N$ to matrix powering over $\Z$.  Unary matrix
powering over $\N$ is $AC^0$-complete for $\#L$ by Claim
\ref{claim:MPcomplnL}. 
Binary matrix powering is reducible to unary matrix powering.

Let $\#STCON^m$ be the problem of $\#STCON$ on multigraphs.  Notice
that $\#STCON^m$ presents an additional layer of difficulty: since
each entry of the adjacency matrix is $\in \N$ rather than $\in \zo$,
it requires more than one bit to store.  This presents a choice of
encodings: the input matrix entries can be encoded in unary or in
binary. 
We will show that the binary version of $\#STCON^m$ is complete for
$\#L$.  This proof provides the necessary insight for Claim
\ref{claim:ZMPcomplnL}. 

\begin{claim} \label{claim:nSTCONmultinL}
Binary $\# STCON^m$ is complete for $\#L$ under $AC^0$-reductions. 
\end{claim}

\begin{proof}
It is obvious that this problem is hard for $\# L$, since $\# STCON$
trivially reduces to $\# STCON^m$.

A few adjustments to the proof of claim \ref{claim:nSTCONcomplnL}
suffice to show that $\#STCON^m \in \# L$.

Again, let $M$ be a logspace Turing machine with specific formatting
of its input.  Let the input graph be represented by a binary string
encoding its adjacency matrix $G$, with $s$ and $t$ the first two
listed vertices. 
Since multiple edges are allowed, this matrix has entries from $\N$.
It is encoded as a binary string by means of the
$\Rowtwo$ function (see Section \ref{notation}), with each matrix
entry $G(u,v)$ written in binary notation as a string for numbers $u$ 
and $v$ referring to vertices.

Notice that there is some maximum $m$ number of edges between any two
vertices in the graph; hence every entry in the adjacency matrix has
at most $log(m)$ bits.  The entire adjacency matrix is encoded in
at least $n^2 \log m$ bits.

As before, $M$ maintains three (binary) numbers on its tape: the
``current'' vertex, the ``next'' vertex, and a count of the number of
edges it has traversed.  The ``current'' vertex is initialized to $s$
(that is, the number $0$, which indexes $s$ in the encoded adjacency
matrix), and the count is initialized to $0$.
$M$ also maintains a bit ``reachable'', which is true iff the ``next''
vertex is reachable from the ``current'' vertex.

When run, $M$ traverses the graph starting at $s$ as follows:
\begin{algorithm}{Traverse-multigraph}{n,s,t,p,G}
\label{algorithm:traverse-multigraph}
  counter \= 0 \\
  current \= s \\
  \begin{WHILE} {counter \leq p \text{and } current \neq t}
    next \= \text{nondeterministically-chosen vertex from} G \\
    reachable \leftarrow 0 \\
    \begin{FOR}{i=|G(current,next)| \TO 0} 
      b \= \text{nondeterministally-chosen bit} \\
      \begin{IF} {reachable=0}
	\begin{IF}{b=0 \text{and} G(current,next)[i]=1}
	  reachable \= 1
	\end{IF} \\
	\begin{IF}{b=1 \text{and} G(current,next)[i]=0}
	  \text{halt and reject}
	\end{IF}
      \end{IF}
    \end{FOR} \\
    \begin{IF}{reachable=0}
     \text{halt and reject}
    \end{IF} \\
    current \= next \\
    counter \= counter + 1
  \end{WHILE} \\
  \begin{IF} {counter > p} \label{foo}
    \text{halt and reject} 
    \ELSE 
    \text{halt and accept}
  \end{IF}
\end{algorithm}

We can think of the $g=G(current,next)$ edges between ``current'' and
``next'' as numbered $0,1, \ldots, g-1$.  The loop on lines 6-12
implicitly selects a $(\log m)$-bit number $n$, one bit at a time,
from most to least significant.  It checks whether the ``next'' vertex
is reachable from the ``current'' vertex along edge numbered $n$.  If
$n<g$, then the ``next'' vertex is reachable from the ``current''
vertex.  The ``reachable'' bit is true if $n<g$ based on the
already-seen bits. 

$M$ simulates a traversal of the graph from $s$ to $t$ by
nondeterministically picking the next edge it traverses and the next
vertex it visits.  Every
accepting computation of $M$ traces a path from $s$ to $t$ (of length
$\leq p$), and for every path of length $\leq p$ from $s$ to $t$ there
is an accepting computation of $M$.  Thus $\# STCON^m \in \# L$.
\end{proof}

Claim \ref{claim:nSTCONmultinL} provides the insight for 
showing that
matrix powering over $\Z$ is $AC^0$-reducible to $\#L$.  In
particular, it presents a technique whereby, using only two bits, a
Turing machine can nondeterministically ``pick'' a natural number
$<n$, where $n$ is given in binary notation.  Notice that this 
has the effect of causing the Turing machine to branch into $n$ 
computational paths.  

\begin{remark} \label{remark:TM-branching}
Branching into $n_1$ paths, then $n_2$ paths, \ldots, then $n_k$ 
paths has the effect of multiplication, resulting in 
$\prod_i n_i$ total computational paths.
\end{remark}

Remark \ref{remark:TM-branching} and the proof of Claim \ref{claim:ZMPcomplnL} are inspired
by the work of Vinay \cite{Vinay}.  Lemma 6.2 of that paper proves a similar fact,
though the conceptual framework, motivation, and notation are different.

\begin{claim} \label{claim:ZMPcomplnL}
Binary matrix powering over $\Z$ is $AC^0$-reducible to $\#L$.
\end{claim}

\begin{proof}
The main idea of this proof is a combination of the reduction of
$\#STCON$ to $\#L$ (Claim \ref{claim:nSTCONcomplnL}) and matrix
powering over $\N$ to $\#STCON$ (Lemma \ref{lemma:MP<=nSTCON}), 
in order to show 
that matrix powering over $\Z$ can be defined entry-by-entry.  We use
the technique from Claim \ref{claim:nSTCONmultinL} to handle the
binary encoding used for integers.

Let $A$ be a matrix of integers, and let string $X$ encode $A$ 
via the $\Rowtwo$ function.
Ignoring the signs of integers, we can interpret $A$ as the adjacency
matrix of a directed multigraph. 
By definition, to show that matrix powering over $\Z$ is
$AC^0$-reducible to $\#L$, we require a $\Sigma_0^B$ formula for the
bit graph of $\PowZ(n,k,X)$.
We will demonstrate a stronger statement: in fact, $\PowZ(n,k,X)$
can be defined by whole entries $A^{k}[i,j]$.

We construct two Turing machines $M^+$ and $M^-$ such that, given a
particular entry $(i,j)$ and integer matrix $X$ as input, both
simulate a traversal of the multigraph that $X$ implicitly represents.
For 
this purpose, the sign is ignored, so that
$\intsize(X^{[i][j]})$ is the number of edges of the multigraph
from vertex $i$ to vertex $j$.  $M^+$ and $M^-$ will use the signs of
traversed edges to decide whether to accept or reject.
The number of accepting paths of $M^+$ minus the number of accepting 
paths of $M^-$ is the value of the $(i,j)^\textrm{th}$ entry of
the matrix product $A^k$.

The machines $M^+$ and $M^-$ have identical instructions
except for their conditions for entering an accepting state.  
Each machine keeps 
track of two numbers, ``current'' and ``next'', indicating the index
($0$, \ldots, $n-1$) of its simulated traversal, the ``count'' of how
many edges it 
has traveled, and a single bit indicating the ``sign'' of the path
traversed so far (when considered as the product of the signs of the
edges traversed).
The number ``current'' is initialized to $i$, ``count'' is initialized
to $0$, and ``sign'' is initialized to ``+''.
When the machine needs to branch into $b$ branches, where $b$ is a
number given in binary, we write ``branch into $b$ paths'' in lieu of
repeating the subroutine given in the {\sc Traverse-multigraph}
algorithm on page \pageref{algorithm:traverse-multigraph}. 

\begin{algorithm}{Positive-matrix-product}{i,j,k,X}
  current \= i \\
  sign \= 0 \\
  count \= 0 \\
  \begin{WHILE} {counter < k}
    next \= \text{nondeterministically-chosen number} < n \\
    \text{branch into $\intsize(X^{[current][next]})$ paths} \\
    sign \= sign \text{ XOR } X^{[current][next]}(0) \\
    count \= count + 1
  \end{WHILE} \\
  \begin{IF} {current = j \text{and } sign = 0}
    \text{halt and accept} 
    \ELSE \text{halt and reject}
  \end{IF}
\end{algorithm}

Given an input $(i,j,k,X)$, $M^+$ traverses the implicit graph
represented by $X$ according to the algorithm {\sc Positive-matrix-product}.
The analogous algorithm {\sc Negative-matrix-product} for $M^-$ will
be identical, except that line $9$ will require that $sign =1$.

Observe that {\sc Positive-matrix-product} is simply an adapted
version of {\sc Traverse-multigraph}.  Rather than traversing a path
between fixed nodes $s=0$ and $t=1$ of any length $\leq p$, it starts
at vertex $i$ and travels across \emph{exactly} $k$ edges.  If the
final vertex of that 
path is $j$, then we have traversed an $i$--$j$ path of length exactly
$k$. 
Remark \ref{remark:matrix-observation} (which provided the insight for
Lemma \ref{lemma:MP<=nSTCON}) and Remark \ref{remark:TM-branching}
complete the proof.

Let $f_{M^+}$ and $f_{M^-}$ be functions of $\#L$, defined as the
number of accepting paths of $M^+$ and $M^-$, respectively.
The number $f_{M^+}(i,j,k,X)$ of accepting computations of $M^+$ is
exactly the number of 
``positive-sign'' paths from $i$ to $j$, i.e., the sum of all positive
terms in the computation of $A^k[i,j]$.
The number $f_{M^-}(i,j,k,X)$ of accepting computations of $M^-$ is
exactly the number of 
``negative-sign'' paths from $i$ to $j$, i.e., the sum of all negative
terms in the computation of $A^k[i,j]$.
Thus the $(i,j)^\textrm{th}$ entry of $A^k$ is given by
$$ A^k[i,j] = f_{M^+}(i,j,k,X) - f_{M^-}(i,j,k,X)$$
\end{proof}

\subsection{The theory $V\#L$} \label{VnL} 

By the nature of its construction, the theory $V\#L$ corresponds to
the $AC^0$-closure of $\#L$.  As we will prove, the set of provably
total functions of $V\#L$ exactly coincides with the functions of
$F\#L$, and the $\Delta_1^B$-definable relations of $V\#L$ are exactly
the relations in $\#L$.

Let $\PowZ(n,k,X)$ and $\PowSeqZ(n,k,X)$ be defined as above
(Definitions (\ref{def:Pow}) and (\ref{def:PowSeq})), with a few
modifications.
Input string $X$ now encodes a matrix, and the outputs encode a list
of matrices.  
(Both encodings are accomplished via the $\Rowtwo$ function.)
   Clearly there are analogs of
Lemmas \ref{lemma:Pow<=PowSeq} and \ref{lemma:PowSeq<=Pow}; 
$\PowZ$ and $\PowSeqZ$ are $AC^0$-reducible to each other.

\begin{definition} \label{def:VnL}
The theory $V\#L$ has vocabulary $\mathcal{L}_A^2$ and is axiomatized
by $VTC^0$ and a $\Sigma_1^B(\mathcal{L}_A^2)$ axiom $PS_\Z$ (formula
\ref{formula:PSz}) stating
the existence of a string value for the function $\PowSeqZ(n,k,X)$.
\end{definition}

We define our new axiom via a series of ``helper'' functions, just as
(\ref{equation:implicit-PS}) in Section \ref{subsection:implicit}.
Let $\delta_{\PowSeqZ}(n,k,X,Y)$ be the relation representing the
graph of $\PowSeqZ(n,k,X) = Y$.  This relation will be defined below
in the language $\mathcal{L}_{FTC^0} \supset \mathcal{L}_A^2$.  This
method requires the introduction of new function symbols, which can be
used to express the axiom $PS_\Z$ in $\overline{VTC^0}$, a universal
conservative extension of $VTC^0$.

Let $ID_\Z(n)=Y$ be the string function whose output encodes the
$n\times n$ identity 
matrix over $\Z$, that is, $Y^{[i][j]}(c) \leftrightarrow i=j \wedge
c=1 $.  The extra layer of encoding necessary for integers makes this
definition inelegant.  (Notice that the previous equation 
defines the $\langle i, \langle j,c\rangle\rangle^\textrm{th}$ bit of
$Y$.) 
\begin{eqnarray*}
 Y(b) & \leftrightarrow & 
\operatorname{\it left}(b) < n \wedge
\operatorname{\it right}(b) < n \wedge
\Pair(b) \wedge  
\Pair(\operatorname{\it right}(b)) \wedge \\
&& \operatorname{\it left}(b) = \operatorname{\it
  left}(\operatorname{\it right}(b)) 
\wedge
\operatorname{\it right}(\operatorname{\it right}(b)) = 1 
\end{eqnarray*}

Let $X_1$ and $X_2$ be two strings encoding $n\times n$ integer
matrices.  Let $G(n,i,j,X_1,X_2)$ be the $TC^0$ string
function that witnesses the computation of the $(i,j)^\textrm{th}$
entry of the matrix product $X_1 X_2$, defined as:
$$ G(n,i,j,X_1,X_2)(b) \leftrightarrow
b< \langle |X_1|,|X_2|\rangle \wedge
\Pair(b) \wedge 
\big (X_1^{[i][\operatorname{\it left}(b)]} \times_\Z 
X_2^{[\operatorname{\it left}(b)][j]} \big)
(\operatorname{\it right}(b))
$$
Here, the bound $b< \langle |X_1|,|X_2|\rangle$ is much larger than
necessary.
The function $G$ serves as a witness; its output is a string encoding
a list of integers $Y_1$, $Y_2$, \ldots, $Y_n$, where
$Y_k = X_1^{[i][k]} \times_\Z X_2^{[k][j]}$.  The above definition
specifies this exactly, as
$$ G(n,i,j,X_1,X_2)^{[\ell]} = X_1^{[i][\ell]} \times_\Z X_2^{[\ell][j]} $$

Thus the $(i,j)^\textrm{th}$ entry of the matrix product $X_1 X_2$ is
given by the sum $Y_1 +_\Z \ldots +_\Z Y_n$.
Chapter 9 of \cite{CookNguyen} defines the string function
$\Sum(n,m,Z)$ that takes the sum of a list of $n$ binary numbers of
length $\leq m$ stored as the first $n$ rows of string $Z$.  We adapt
this definition to create string function $\Sumz$, which performs the same operaton on a list of $n$ integers.

\newcommand{\PosList}{\operatorname{\it PosList}}
\newcommand{\NegList}{\operatorname{\it NegList}}
\newcommand{\PosSum}{\operatorname{\it PosSum}}
\newcommand{\NegSum}{\operatorname{\it NegSum}}
As a first step, we partition the list $Z$ of integers into two lists: positive and negative numbers.
Within each of these partitions, integers are stored as their $\BinaryPart$, i.e., without a leading sign.
Define the string functions $\PosList$ and $\NegList$ as:
$$ \PosList(Z)(i,j) \leftrightarrow \neg Z^{[i]} \wedge \BinaryPart(Z^{[i]})(j) $$
$$ \NegList(Z)(i,j) \leftrightarrow      Z^{[i]} \wedge \BinaryPart(Z^{[i]})(j) $$
The function $\PosList(Z)$ (respectively, $\NegList(Z)$), outputs 
a list of the same length as $Z$, but only including positive (negative) elements of $Z$; 
all other entries are zero.

Thus $\Sum(n,m,\PosList(Z))$ is the natural number encoding the sum of the positive integers in list $Z$, 
and $\Sum(n,m,\NegList(Z))$ is the natural number encoding the \emph{negation} of the sum of the
negative integers in list $Z$.
The output of each of these functions is a natural number.
It is necessary to ``shift" the bits 
by one place to re-insert the sign required by our encoding scheme for integers.
$$ \PosSum(n,m,Z)(b) \leftrightarrow \exists k<b, k+1=b \wedge \Sum(n,m,\PosList(Z))(k)$$
$$ \NegSum(n,m,Z)(b) \leftrightarrow (b=0) \vee (\exists k<b, k+1=b \wedge \Sum(n,m,\NegList(Z))(b))$$
This allows for a simple definition of integer summation $\Sumz$:
$$ \Sumz(n,m,Z)(i,j) \leftrightarrow  = \PosSum(n,m,Z) +_\Z \NegSum(n,m,Z) $$

The string function $\ProdZ$ computing the product of two integer
matrices can be bit-defined as:
$$ \ProdZ(n,X_1, X_2)(i,j) \leftrightarrow
i<n \wedge j<n \wedge
\Sumz(n,|X_1|+|X_2|,G(n,i,j,X_1,X_2)) $$
Again, the bound $|X_1|+|X_2|$ is much larger than necessary.

Given these functions, let the relation $\delta_{\PowSeqZ}(n,k,X,Y)$
be the the $\Sigma_0^B(\mathcal{L}_{FTC^0})$-formula
\begin{eqnarray}
  \forall b<|Y|, 
  |Y| < \langle k, \langle |X|, |X| \rangle \wedge
  [ Y(b) \supset (\Pair(b) \wedge \Pair(\operatorname{\it right}(b)))]
  \wedge 
  \nonumber 
  \hspace{1in} \\
  \hspace{2in}
  Y^{[0]} = ID_\Z(n) \wedge
  \forall i<k (Y^{[i+1]} = \ProdZ (n,X,Y^{[i]}))
  \label{formula:implicit-delta-PSz}
\end{eqnarray}
This formula asserts that the string $Y$ is the output of
$\PowSeqZ(n,k,X)$.  By our convention for defining string functions,
the bits of $Y$ that do not encode the list  $[X^1,\ldots,X^k]$ are
required to be zero, so $\PowSeqZ(n,k,X)$ is the lexographically
first string that encodes this list.

$\overline{VTC^0}$ is a conservative extension of $VTC^0$, defined in
Section 9C of \cite{CookNguyen}.
Theorem 9.33(b) of \cite{CookNguyen} asserts that there is
a $\mathcal{L}_A^2$ term $t$ and a $\Sigma_0^B(\mathcal{L}_A^2)$
formula $\alpha_{\PowSeqZ}$ such that
$$ \exists Z < t, \alpha_{\PowSeqZ}(n,k,X,Y,Z)$$
is provably equivalent to (\ref{formula:implicit-delta-PSz}) in
$\overline{VTC^0}$. 

The axiom $PS_\Z$ used to define the theory $V\#L$ is
\begin{equation}
\label{formula:PSz}
\exists Y<m, \exists Z<t, \alpha_{\PowSeqZ}(n,k,X,Y,Z)
\end{equation}

Formulas (\ref{formula:implicit-delta-PSz}) and (\ref{formula:PSz})
each require that $Y$ witnesses the intermediate strings $X^1$, $X^2$,
\ldots, $X^k$ of the computation of the matrix power $X^k$.  String
$Y$ does \emph{not} witness any of the work performed in calculating
these intermediate powers of $X$, just as in Section
\ref{subsection:implicit}.  That work, as before, is witnessed by the
string $Z$, although this witnessing is obscured by the application of
Theorem 9.33(b). 

The next two lemmas are analogous to Lemmas
\ref{lemma:PowSeq_definable} and \ref{lemma:PowSeq*_definable}, and
proved in the same manner.

\begin{lemma}
\label{lemma:PowSeqz_definable}
The integer matrix powering function $\PowSeqZ$ is
$\Sigma_1^B(\mathcal{L}_A^2)$-definable in $V\#L$.
\end{lemma}

\begin{lemma}
\label{lemma:Powseq*z_definable}
The aggregate integer matrix powering function $\PowSeqZstar$ is
$\Sigma_1^B(\mathcal{L}_A^2)$-definable in $V\#L$.
\end{lemma}

\subsection{The theory $\overline{V\#L}$}
This section defines the theory $\overline{V\#L}$, a universal
conservative extension of $V\#L$, in the same way as
$\overline{V\parityL} \supset V\parityL$ (Section
\ref{subsection:overline-V2L}). 

The theory $\overline{V\#L}$ is a universal conservative extension of
$V\#L$.  Its language $\mathcal{L}_{F\#L}$ contains function symbols
for all string functions in $F\#L$.  Note that $F\#L \neq \#L$, which
can be observed from the fact that $F\#L$ contains string functions
and number functions that take negative values, neither of which is in
$\#L$. 
The defining axioms for the functions in $\mathcal{L}_{F\#L}$ are
based on their $AC^0$ reductions to matrix powering.  Additionally,
$\overline{V\#L}$ has a quantifier-free defining axiom for
$\PowSeqZ'$, a string function with inputs and outputs the same as
$\PowSeqZ$. 

As before, we leave the formal definition of $\PowSeqZ$ unchanged.
$\PowSeqZ'$ has the quantifier-free defining axiom
\begin{eqnarray}
  |Y| < \langle k, \langle |X|, |X| \rangle \wedge
  ( Y(b) \supset \Pair(b) \wedge \Pair(\operatorname{\it right}(b)))
  \wedge 
  \nonumber 
  \hspace{1in} \\
  \hspace{2in}
  Y^{[0]} = ID_\Z(n) \wedge
  \big( i<k \supset (Y^{[i+1]} = \ProdZ (n,X,Y^{[i]})) \big)
  \label{formula:quant-free-defz}
\end{eqnarray}
This formula is similar to (\ref{formula:implicit-delta-PSz}).  The
function $\PowSeqZ$ satisfies this defining axiom for
$\PowSeqZ'$.  $VTC^0$, together with both axioms, proves
$\PowSeqZ(n,k,X) = \PowSeqZ'(n,k,X)$.

Recall that for a given formula $\varphi(z,\vec x,
\vec X)$ and $\mathcal{L}_{A}^2$-term $t(\vec x, \vec X)$, we let
$F_{\varphi,t}(\vec x, \vec X)$ be the string function with bit
definition (\ref{formula:FvarphiDefAxiom})
$$   Y = F_{\varphi(z),t}(\vec x, \vec X) \leftrightarrow
  z < t(\vec x, \vec X) \wedge \varphi(z,\vec x, \vec X)
$$

\begin{definition}[$\mathcal{L}_{F\#L}$]
We inductively define the language $\mathcal{L}_{F\#L}$ of all
functions with (bit) graphs in $\#L$.
Let $\mathcal{L}_{F\#L}^0 = \mathcal{L}_{FTC^0}\cup
\{\PowSeqZ'\}$.  
Let $\varphi(z,\vec x,\vec X)$ be an open formula over
$\mathcal{L}_{F\#L}^i$, and let $t=t(\vec x, \vec X)$ be a
$\mathcal{L}_A^2$-term.  
Then $\mathcal{L}_{F\#L}^{i+1}$ is $\mathcal{L}_{F\#L}^i$
together with the string function $F_{\varphi(z),t}$ that has defining 
axiom:
\begin{equation} \label{equation:F_varphiz}
F_{\varphi(z),t}(\vec x, \vec X)(z) \leftrightarrow z < t(\vec x, \vec
X) \wedge \varphi(z,\vec x, \vec X)
\end{equation}

Let $\mathcal{L}_{F\#L}$ be the union of the languages
$\mathcal{L}_{F\#L}^i$.  Thus $\mathcal{L}_{F\#L}$ is the
smallest set containing $\mathcal{L}_{FTC^0}\cup\{\PowSeqZ'\}$ and
with the defining axioms for the functions $F_{\varphi(z),t}$ for
every open $\mathcal{L}_{F\parityL}$ formula $\varphi(z)$.
\end{definition}

As before, $\mathcal{L}_{F\#L}$ has a symbol for every string
function in $F\#L$, and a term $|F|$ for every number function in
$F\#L$, where $F$ is a string function in $F\#L$.

\begin{definition}[$\overline{V\#L}$] The universal theory
$\overline{V\#L}$ over language $\mathcal{L}_{F\#L}$ has the
axioms of  
$\overline{VTC^0}$ together with the quantifier-free defining axiom 
(\ref{formula:quant-free-defz}) for $\PowSeqZ'$ and the above defining
axioms (\ref{equation:F_varphiz}) for the functions
$F_{\varphi(z),t}$. 
\end{definition}

\begin{theorem} \label{theorem:overline-UCE-nL}
$\overline{V\#L}$ is a universal conservative extension of
  $V\#L$.
\end{theorem}

The proof is the same as for Theorem
\ref{theorem:overline-UCE-2L} (page
\pageref{theorem:overline-UCE-2L}).

\subsection{Provably total functions of $V\#L$}
\label{section:provably-total-VnL}

These results follow directly from \cite{CookNguyen}.  We restate them
here for convenience.

\begin{claim}
The theory $\overline{V\#L}$ proves the axiom schemes
$\Sigma_0^B(\mathcal{L}_{FTC^0}) \comp$, \\
$\Sigma_0^B(\mathcal{L}_{FTC^0}) \ind$, and
$\Sigma_0^B(\mathcal{L}_{FTC^0}) \minaxiom$.
\end{claim}

\begin{claim}
\begin{enumerate}
  \renewcommand{\labelenumi}{(\alph{enumi})}
  \item A string function is in $F\#L$ if and only if it is represented
    by a string function symbol in $\mathcal{L}_{F\#L}$.
   \item A relation is in $\#L$ if and only if it is represented by
    an open formula of $\mathcal{L}_{F\#L}$ if and only if it is
    represented by a $\Sigma_0^B(\mathcal{L}_{F\#L})$ formula.
 \end{enumerate}
\end{claim}

\begin{corollary}
Every $\Sigma_1^B(\mathcal{L}_{F\#L})$ formula $\varphi^+$ is
equivalent in $\overline{V\#L}$ to a $\Sigma_1^B(\mathcal{L}_A^2)$
formula $\varphi$.
\end{corollary}

\begin{corollary} \label{cor:definability-nL}
\begin{enumerate}
    \renewcommand{\labelenumi}{(\alph{enumi})}
    \item A function is in $F\#L$ iff it is
      $\Sigma_1^B(\mathcal{L}_{F\#L})$-definable in
      $\overline{V\#L}$ 
      iff it is $\Sigma_1^B(\mathcal{L}_{A}^2)$-definable in
      $\overline{V\#L}$.
    \item A relation is in $\#L$ iff it is
      $\Delta_1^B(\mathcal{L}_{F\#L})$-definable in
      $\overline{V\#L}$ iff it is
      $\Delta_1^B(\mathcal{L}_A^2)$-definable in $\overline{V\#L}$.
\end{enumerate}
\end{corollary}


\begin{theorem} \label{theorem:provablytotalFnL}
A function is provably total ($\Sigma_1^1$-definable) in $V\#L$ iff
it is in $F\#L$.
\end{theorem}

\begin{theorem}
  A relation is in $\#L$ iff it is
  $\Delta_1^B$-definable in $V\#L$ iff it is $\Delta_1^1$-definable
  in $V\#L$.
\end{theorem} 


%% file: parityLfuturework.tex
\section{Future work} \label{section:future-work}

Due to the time constraints of this project, there are several results
omitted above.  These require work beyond 
what is contained in this paper.

The proof for Theorem \ref{theorem:2L-closed} (page
\pageref{theorem:2L-closed}), stating the $AC^0$-closure of 
$\parityL$, is omitted above, in lieu of which several references are
given for papers proving the same closure, or a stronger version.
For completeness, this theorem should be proven within the framework
of Chapter 9, by use of Theorem 9.7.

Similarly, this work is incomplete without a proof that the closure of
$\#L$ is the class $DET$.  On page \pageref{def:DET} the
class $DET$ is simply defined as this closure; the exposition would be
more self-contained if this characterization of $AC^0(\#L)$ were
proven.  (It is proven or implied by results in  \cite{All04},
\cite{MV}, \cite{AllOg}, and others.)

By using the notation $DET$ for the
closure $AC^0(\#L)$, we implicitly 
rely upon the fact that the problem of computing the determinant of an
integer-valued matrix is complete for $\#L$.  The notation developed
above is sufficient for a proof that the integer
determinant can be captured by reasoning in $V\#L$.  This proof would
likely formalize Berkowitz's method, which provides a reduction from
integer determinant to integer matrix powering \cite{Berkowitz}.

There are a number of algebraic problems which  \cite{Cook},
\cite{BDHM}, \cite{BvGH}, and \cite{Berkowitz} prove are
$NC^1$-reducible to each other.  These include:
 integer determinant,
 matix powering,
 iterated matrix product,
 computing the coefficients of characteristic polynomials,
 rank computation,
 choice of a linearly independent subset from a set of vectors,
 computing the basis of the kernel of a matrix, and
 solving a system of linear equations.
Many of these reductions seem to be simple enough to be formalized as
$AC^0$-reductions, and intuitively it seems that all of these problems
should also be complete for $\parityL$ and $\#L$ under
$AC^0$-reductions. Since $\parityL$ concerns elements over a field, 
finding matrix inverses can be added to the list (and similarly for
every $MOD_pL$ class for $p$ prime).

Formalizing these reductions seems an arduous task, but a shortcut is
possible.  In \cite{SC}, the authors construct formal theories for
linear algebra over three sorts: indices ($\N$), field elements, and
matrices.  
Several basic theorems of linear algebra, including
many of the problems listed 
above, are provable in these theories.
By interpreting these 3-sorted theories into the 2-sorted
theories constructed in this paper, we can use convenient results
without having to prove them again in a different framework.